\documentclass[acmsmall]{acmart}

\AtBeginDocument{%
  }

\usepackage[english]{babel}
\usepackage[utf8]{inputenc}
\usepackage[T1]{fontenc}
\usepackage{wrapfig}
\usepackage{graphicx}

\usepackage{hyperref}
\usepackage{mathtools}
\usepackage{enumerate}

\usepackage{xr}
\usepackage{tikz}
\usetikzlibrary{positioning, decorations.pathreplacing}
\usetikzlibrary{tikzmark}

\usepackage{algpseudocode}
\usepackage{algorithm}
\usepackage{graphicx}
\usepackage[utf8]{inputenc}
\usepackage{textcomp}
\usepackage{xcolor}
\usepackage{soul}
\usetikzlibrary{positioning, arrows.meta}
\usetikzlibrary{positioning, decorations}

\usepackage{caption}

\usepackage[dvipsnames]{xcolor}

\usepackage{subcaption}

\usepackage{enumitem}

\usepackage{amsmath}
\usepackage{verbatim}

\usepackage{wrapfig}

\usetikzlibrary{calc}

\usepackage{caption}
\usepackage{subcaption}

\newtheorem{theorem}{Theorem}
\newtheorem{lemma}{Lemma}
\newtheorem{proposition}{Proposition}
\newtheorem{corollary}{Corollary}

\newtheorem{definition}{Definition}

\newtheorem{assumption}{Assumption}
\newtheorem{specification}{Specification}
\newtheorem*{mtheorem}{Meta-Theorem}
\newtheorem*{massumption}{Meta-Assumption}

\newcommand{\nat}{\mathbb{N}}
\newcommand{\reals}{\mathbb{R}}
\newcommand{\posreals}{\mathbb{R}_+}
\newcommand{\consi}{\beta_t^i}
\newcommand{\feasconsi}{q_t^i}
\newcommand{\cons}{\beta_t}

\newcommand{\utili}{U^i}
\newcommand{\pricet}{\alpha_t}
\newcommand{\asconsi}{\hat{\beta}_t^i}
\newcommand{\dataset}{\mathcal{D}}
\newcommand{\ndataset}{\tilde{\mathcal{D}}}

\newcommand{\CX}{\mathcal{X}}
\newcommand{\CA}{\mathcal{A}}

\newcommand{\gaussN}{\mathcal{N}}

\newcommand{\PR}{\mathbb{P}}
\newcommand{\CE}{\mathbb{E}}

\newcommand{\na}{n}
\newcommand{\numconsts}{N}
\newcommand{\boldf}{\boldsymbol{f}}
\newcommand{\conset}{X_c^{\dtime}}
\newcommand{\effset}{X_E}
\newcommand{\simplex}{\mathcal{W}_{\na}}
\newcommand{\psimplex}{\simplex^+}
\newcommand{\woptset}{S(\sw)}

\newcommand{\nlconst}{g_{\mech,\dtime}}
\newcommand{\dtime}{t}
\newcommand{\resp}{x_{\dtime}}

\newcommand{\cf}{\hat{f}} 
\newcommand{\et}{\eta_t}
\newcommand{\sgcf}{\nabla f^i(\bxt)}
\newcommand{\sgcn}{\nabla \nlconst^i (\bxt)}

\newcommand{\gam}{\Gamma}
\newcommand{\sw}{\mu}

\newcommand{\at}{\textbf{a}_t}
\newcommand{\ba}{\textbf{a}}
\newcommand{\as}{\textbf{a}_s}

\newcommand{\grelin}{R_0^i} 
\newcommand{\greli}{R^i} 

\newcommand{\lconst}{\alpha_{\dtime}}

\newcommand{\del}{\partial}

\newcommand{\var}{\ba}

\newcommand{\dimn}{N}

\newcommand{\CD}{\mathcal{D}}

\newcommand{\ms}{\mu}
\newcommand{\msi}{\mu^i}
\newcommand{\msb}{\ms}

\newcommand{\bxti}{\bar{x}_t^i}
\newcommand{\bxsi}{\bar{x}_s^i}
\newcommand{\bxt}{\bar{x}_t}
\newcommand{\bxs}{\bar{x}_s}

\newcommand{\bg}{\bar{g}}

\newcommand{\wass}{\mathcal{W}}

\newcommand{\CV}{\mathcal{V}}
\newcommand{\bv}{\boldsymbol{v}}

\newcommand{\hu}{\hat{u}}
\newcommand{\hlam}{\hat{\lambda}}

\newcommand{\mech}{\phi}
\newcommand{\mechspace}{\Phi}

\newcommand{\itr}{m}

\DeclareMathOperator{\nashe}{\mathcal{N}}
\DeclareMathOperator{\soco}{\mathcal{S}}

\DeclareMathOperator{\revp}{RP}
\DeclareMathOperator{\rnk}{rank}
\DeclareMathOperator{\pos}{pos}
\DeclareMathOperator{\maxrnk}{maxrank}

\usepackage{tikz}
\usetikzlibrary{positioning}

\usepackage{mdframed}
\usepackage{lipsum}

\usepackage{float}

\begin{document}

\title{Data-Driven Mechanism Design via Multi-Agent Revealed Preferences}

%




\author{Luke Snow}
\email{las474@cornell.edu}
\author{Vikram Krishnamurthy}
\email{vikramk@cornell.edu}
\affiliation{%
  \institution{Cornell University}
  \city{Ithaca}
  \state{New York}
  \country{USA}
}

\begin{abstract}
\textcolor{black}{We study a sequence of independent one-shot non-cooperative games in which agents generate decisions according to equilibria parameterized by a mechanism. Given only these observed decisions, and without parametric knowledge of the agents’ utility functions, our goal is to achieve mechanism design by steering equilibria toward social optimality. When this is infeasible due to the intrinsic structure of the underlying game, we seek to identify such unachievability.} This paper constructs a reinforcement learning (RL) framework for adaptively achieving this mechanism design objective. 
We first derive a novel multi-agent revealed preference test for Pareto optimality -- this yields necessary and sufficient conditions for the existence of utility functions under which empirically observed mixed-strategy Nash equilibria are socially optimal. These conditions take the form of a testable linear program, and this result is of independent interest. We utilize this result to construct an inverse reinforcement learning (IRL) step to determine the  \textit{Pareto gap}, i.e., the distance of observed strategies from Pareto optimality. We pair this IRL step with an RL policy gradient algorithm and prove convergence to a mechanism which globally minimizes the Pareto gap. \textcolor{black}{This enables a principled test for achievability of social optimality: if social optimality is achievable for the game setup and observational equilibria, Algorithm~\ref{alg:amd} will achieve it; else, we can \textit{identify} this unachievability, while still converging to a mechanism which is closest to social optimality}. We also reveal an intimate connection between our constructed loss function and several robust revealed preference metrics; this allows us to reason about algorithmic suboptimality through the lens of these well-established microeconomic principles. Finally, in the case when only finitely many i.i.d. samples from mixed-strategies (partial strategy specifications) are available, we derive concentration bounds for our algorithm's convergence, and we construct a distributionally robust RL procedure which achieves \textcolor{black}{the mechanism design objective} for the fully specified strategies.\footnote{A preliminary version of this work appears in the 2024 IEEE Conference on Decision and Control. The current paper substantially  extends  the results to the mixed-strategy regime, reveals connections to well-established relaxed revealed preference metrics, and develops a distributionally robust optimization procedure for handling partially specified strategies.}
\footnote{This research was supported by NSF grants CCF-2312198 and CCF-2112457 and U. S. Army Research Office under grant W911NF-24-1-0083}
\end{abstract}


\maketitle

\newpage 
\section{Introduction}

We consider a sequence of independent one-shot non-cooperative games. Multiple agents act in their own self-interest, aiming to maximize their individual utility functions which are dependent on the actions of the entire set of agents. Furthermore, agents are allowed to randomize over actions in order to play from mixed-strategies. It is well known that mixed-strategy Nash equilibria represent the stable joint-strategy states of the entire group, consisting of those states in which no agent has any incentive to unilaterally deviate from its mixed-strategy given the strategies of all other agents. 

The goal of \textit{mechanism design} is to fashion the game structure, specifically the mapping from joint-action space to agent utilities ("mechanism"), such that \textit{non-cooperative behavior gives rise to socially optimal solutions}. \textcolor{black}{The \textit{mechanism} specifies how agents’ chosen actions determine the resulting outcomes and thus individual utilities—it is the rule that translates strategic behavior into payoffs. Thus, by modifying the mechanism, the strategic agent outputs will adjust accordingly, and a mechanism designer can aim to indirectly adjust these outputs through the mechanism, such that they become socially optimal.}


Existing mechanism design approaches exploit observations \textcolor{black}{or reports} of the agent utility functions to construct mechanisms which induce socially optimal outcomes from non-cooperative interaction. In this work we develop a data-driven reinforcement learning (RL) approach to achieve mechanism design with \textcolor{black}{no analytical, parametric, or distributional knowledge of} agent utilities (Our terminology of RL is consistent with the machine learning literature, where the aim is to achieve data-driven optimization without direct knowledge or estimation of the system parameters.) 

\subsection{Main Result. Data-Driven Mechanism Design}
We first introduce some notation. 
Consider a sequence of independent one-shot games comprising $\na$ agents.
$\mech \in \mechspace$ denotes a vector that parametrizes a mechanism. $\mu(\mech)$ denotes a mixed-strategy with respect to the mechanism $\mech$.  $\nashe(\mech)$ denotes the set of mixed-strategy Nash equilibria with respect to $\mech$. Finally,  $\soco(\mech)$ denotes the set of socially optimal strategies with respect to $\mech$, i.e., where Pareto optimality is achieved. Note that socially optimal strategies constitute a subset of Pareto optimal strategies.

\textcolor{black}{The following assumptions characterize the main information structure in this paper. The precise technical assumptions are in Section~\ref{sec:confungen} and \ref{sec:as}.}

{\color{black}\begin{massumption}
The mechanism designer operates under the following informational regime:
\begin{enumerate}[label=\roman*)]
    \item \textit{Non-parametric utility framework:} The designer has no analytical, parametric, or distributional knowledge of agents' utility functions, and elicits no self-reported valuations. We assume the utility functions are concave, monotone and differentiable. 
    \item \textit{Behavioral observability:} Only joint strategy profiles (possibly mixed strategies) are observed for each probed mechanism parameter $\mech \in \mechspace$, where $\mechspace$ is a compact subset of  $\reals^p$. We also consider the finite-sample case where empirical distributions are observed that represent noisy versions of the joint strategy profiles. 
    \item \emph{Equilibrium-consistent agents (operational).}
    The agents play from mixed-strategy Nash equilibria.\footnote{This is the standard primitive of mechanism design and is adopted here as an operational hypothesis rather than a technical requirement for our results; see the discussion in Section~\ref{sec:motddrl}. Thus, our results also apply to other behavioral protocols (e.g., Bayes--Nash or quantal-response equilibrium).}.
\end{enumerate}
\label{massumption}
\end{massumption}

Meta-assumption~\textit{(iii)} is the classical mechanism-design hypothesis and is used interpretively, as made precise in Section~\ref{sec:motddrl}. Meta-assumptions~\textit{(i)} and~\textit{(ii)} are non-standard, and motivate our main result. The assumption of concave utility functions in meta-assumption $i)$ is widely used to ensure existence of Nash equilibria, see \cite{rosen1965existence} for example. The remaining technical assumptions, imposing sufficient smoothness and domain compactness, will imply the existence of a mechanism admitting a socially optimal mixed-strategy Nash equilibrium. This makes our objective feasible: to steer observed mixed-strategy Nash equilibria towards social optimality by adapting the mechanism. The following meta-Theorem provides a succinct statement of our main results.

\begin{mtheorem}[Informal Statement of Results]\leavevmode
\begin{enumerate}[label=\roman*)]
    \item (Generalization of Afriat's Theorem) A necessary and sufficient condition for observed mixed-strategies to be socially optimal, is that a set of linear inequalities has a feasible solution (Theorem~\ref{thm:MA_Af}). 
    \item (Pareto Gap) Let the non-negative Pareto gap $L(\msb(\mech))$ denote how far from feasibility the linear inequalities in Statement i) are. Then, Algorithm~\ref{alg:amd} (in Section~\ref{sec:spsa}) converges in probability to $\arg\min_{\mech\in\mechspace}L(\msb(\mech))$. This guarantees that one of the following holds, for Nash equilibria $\msb(\mech)$:
    \begin{enumerate}[label=\arabic*)]
    \item \emph{(Achieving Social Optimality).} If the global minimum of the Pareto gap $L(\msb(\mech))$ is zero, then a mechanism $\mech^*$ exists under which $\msb(\mech^*)$ is socially optimal, and furthermore Algorithm~\ref{alg:amd} converges in probability to $\mech^*$ (inducing $\mu(\mech^*) \in \nashe(\mech^*) \Rightarrow \mu(\mech^*)\in \soco(\mech^*))$ \\ \textbf{or} 
    \item \emph{(Impossibility of Achieving Social Optimality).} If the global minimum of $L(\msb(\mech))$ is strictly greater than zero, then no mechanism inducing social optimality exists. Algorithm~\ref{alg:amd} produces a mechanism which is closest to social optimality (as quantified by robustness measures in  statement $iii)$)
    \end{enumerate}
    \item (Robustness Measures) The robust revealed preference metrics such as Critical Cost Efficiency Index (CCEI) and $GARP_F$ coincide with the Pareto gap $L(\msb(\mech))$ metric (Theorem~\ref{thm:rrp}), thus allowing us to quantify sub-Pareto-optimality ($L(\msb(\mech)) > 0)$ through the lens of these principles.
    \item (Finite-Sample Robustness) Given finite-sample empirical distributions within a specified Wasserstein distance to the true joint strategies, an optimal mechanism $\mech^* \in\arg\min_{\mech\in\mechspace}L(\msb(\mech))$ can be obtained using distributionally robust optimization (DRO) (Theorem~\ref{thm:dreq}); however, this requires solving a semi-infinite program (Theorem~\ref{thm:sireform}). We provide a finite-dimensional algorithm (finite-exchange method) to construct $\delta$-optimal solutions to this DRO (Theorem~\ref{thm:alg2convg}). 
\end{enumerate}
\end{mtheorem}}

Thus, we can \emph{identify} whether Pareto-optimality is a \emph{possible} outcome for a designer, and if so Algorithm~\ref{alg:amd} automatically finds the mechanism inducing it.
The upshot of the meta-Theorem is a novel perspective for mechanism design, driven by the capabilities of revealed preference methodologies for detecting optimal microeconomic behavior. 

\begin{figure}[h]
\centering

\tikzset{every picture/.style={line width=0.75pt}} 

\begin{tikzpicture}[x=0.75pt,y=0.75pt,yscale=-1,xscale=1]

\draw   (88.5,101) -- (286.5,101) -- (286.5,150) -- (88.5,150) -- cycle ;
\draw    (489.5,68) -- (489.97,98) ;
\draw [shift={(490,100)}, rotate = 269.1] [color={rgb, 255:red, 0; green, 0; blue, 0 }  ][line width=0.75]    (10.93,-3.29) .. controls (6.95,-1.4) and (3.31,-0.3) .. (0,0) .. controls (3.31,0.3) and (6.95,1.4) .. (10.93,3.29)   ;
\draw    (119.5,67) -- (489.5,67) ;
\draw    (119.5,67) -- (119.5,101) ;
\draw    (490.5,150) -- (490.5,196) ;
\draw    (122.5,196) -- (122.5,165) -- (122.5,153) ;
\draw [shift={(122.5,151)}, rotate = 90] [color={rgb, 255:red, 0; green, 0; blue, 0 }  ][line width=0.75]    (10.93,-3.29) .. controls (6.95,-1.4) and (3.31,-0.3) .. (0,0) .. controls (3.31,0.3) and (6.95,1.4) .. (10.93,3.29)   ;
\draw   (339.5,100) -- (522.5,100) -- (522.5,150) -- (339.5,150) -- cycle ;
\draw   (257.5,174) -- (390.5,174) -- (390.5,225) -- (257.5,225) -- cycle ;
\draw    (490.5,197) -- (391.5,197) ;
\draw    (255.5,196) -- (228.5,196) -- (122.5,196) ;
\draw [dotted]  (195,40.5) -- (455,40.5) -- (455,63.5) -- (195,63.5) -- cycle ;

\draw (101,109) node [anchor=north west][inner sep=0.75pt]   [align=left] {\textit{Mechanism Evaluation}:\\{\small IRL via Revealed Preferences}};
\draw (352,108) node [anchor=north west][inner sep=0.75pt]   [align=left] {\textit{Mechanism Adjustment}:\\{\small Policy gradient optimization}};
\draw (259.5,177) node [anchor=north west][inner sep=0.75pt]   [align=left] {\textit{Black-Box:}\\{\small multi-agent system}\\{\small with unknown utilities}};
\draw (415,160) node [anchor=north west][inner sep=0.75pt]   [align=left] {Mechanism \\ parameter $\displaystyle \mech _{k}$};
\draw (130,160) node [anchor=north west][inner sep=0.75pt]   [align=left] {Mixed strategy \\ $\displaystyle \mu ( \mech _{k}) \ \in \nashe( \mech _{k})$};
\draw (123,75) node [anchor=north west][inner sep=0.75pt]   [align=left] {Pareto gap $\displaystyle L( \mu( \mech _{k}))$};
\draw (220,42.5) node [anchor=north west][inner sep=0.75pt]   [align=left] {$\displaystyle L( \mu ( \mech _{k})) \ =\ 0\ \Leftrightarrow \ \mu ( \mech _{k}) \ \in \ \soco( \mech _{k})$};

\end{tikzpicture}
\caption{The reinforcement learning (RL) framework for achieving mechanism design comprises two steps: mechanism evaluation and mechanism adjustment. The mechanism evaluation step utilizes revealed preference theory as a form of inverse reinforcement learning (IRL), to determine the Pareto gap $L(\mu(\mech_k))$, defined in Section~\ref{sec:parreg}. This measures the proximity of observed mixed strategy equilibria $\mu(\mech_k) \in \nashe(\mech_k)$ to social optimality $\soco(\mech_k)$. The mechanism adjustment phase utilizes this IRL metric to perform policy optimization, by updating the mechanism according to a simultaneous perturbation stochastic approximation (SPSA) algorithm. Thus, this procedure operates without analytical or observational knowledge of the agent utility functions. The novelty lies in our IRL approach to policy evaluation, which generalizes results in the theory of microeconomic revealed preferences.}\label{fig:mflow}
\end{figure}


\subsection{Motivation. Data-Driven Reinforcement Learning (RL) Framework}
\label{sec:motddrl}
{\color{black} The procedure outlined in statements $i)$ and $ii)$ of the Meta-Theorem can be viewed as a novel RL framework that achieves the mechanism design objective. 
This RL procedure operates by first measuring the Pareto gap for an evaluated mechanism, then adjusting the mechanism in the direction of decreasing Pareto gap. 
Figure~\ref{fig:mflow} decomposes this RL procedure into three key steps:
\begin{enumerate}[label=\roman*)]
    \item \textit{Agent-Response Model.} 
    For a fixed mechanism parameter $\mech \in \mechspace$, each agent plays from a resultant mixed strategy. This module maps mechanism parameter $\mech$ to the observed (equilibrium) distribution $\mu(\mech)$. 

    \item \textit{Pareto-Gap Evaluation.}
    The designer quantifies the proximity $L(\mu(\mech))$ of $\mu(\mech)$ to social optimality (Pareto gap) by solving an Afriat-style linear program of generalized revealed preference inequalities. This module is an implicit inverse reinforcement learning (IRL)\footnote{In machine learning, IRL \cite{ng2000algorithms} is a form of inverse optimization whereby an observed agent's cost function is recovered from observations of its behavior. In the microeconomic literature IRL \cite{krishnamurthy2020identifying} aims to first detect utility maximization behavior.} step.

    \item \textit{Mechanism-Update Rule.}
    The designer adjusts the mechanism parameters using a stochastic gradient-free optimizer such as simultaneous perturbation stochastic approximation (SPSA), evolutionary strategies, or Bayesian optimization. This update corresponds to the \textit{policy-improvement} step, iteratively refining $\mech$ toward mechanisms that induce equilibrium play consistent with Pareto optimality.
\end{enumerate}


Statements $iii)$ and $iv)$ of the meta-Theorem generalize the above RL procedure to achieve \textit{robust} RL for mechanism design, allowing a principled quantification of algorithmic sub-optimality and \textit{finite-sample} performance. To summarize, this framework allows us to achieve mechanism design \textcolor{black}{through only mixed strategy observations} by first quantifying, through IRL, the proximity between observed strategies and social optima, and second to optimize the mechanism by an iterative stochastic gradient-based procedure.\footnote{The observed strategies need not be mixed-strategy Nash equilibria for Algorithm~\ref{alg:amd} to be well-defined; the role of this hypothesis is interpretive and is discussed under \emph{Operational mixed-strategy equilibria} below.}}

\paragraph{Information Asymmetry.}
\textcolor{black}{Classical mechanism design approaches such as Vickrey--Clarke--Groves (VCG)
assume access to agents’ self--reported valuations, possibly subject to
strategic misreporting.  In contrast, many real--world systems reveal only
agents’ \emph{actions} while utilities remain latent.
Examples include bidding behavior in ad auctions
\cite{varian2007position,edelman2007internet},
power or channel selection in wireless networks
\cite{niyato2008dynamics},
and distributed coordination among UAVs or sensors
\cite{marden2009cooperative}.
In such settings, agents typically possess enough mutual knowledge or feedback
to randomize optimally given others’ behavior—yielding empirically observed
mixed--strategy or quantal--response equilibria
\cite{mckelvey1995quantal}.
The designer, however, observes only these realized strategies,
not the underlying utilities.
Our framework is thus motivated by this practical information asymmetry:
we seek to infer or design mechanisms directly from observed equilibrium play,
without requiring explicit utilities or self-reported valuations.} \textcolor{black}{While there do exist techniques, e.g., iterative, price-based mechanism design that  uses only demand queries and not full valuations, we motivate in Section~\ref{sec:mdsc} the distinguishing generality and applicability of our approach.} 

\paragraph{Operational Mixed-Strategy Equilibria} 
\textcolor{black}{We distinguish the technical content of the algorithm from its mechanism-design interpretation. Theorems~\ref{thm:MA_Af} and~\ref{thm:convg}, together with Algorithm~1, take as input a dataset of observed mixed strategies and return a mechanism minimizing the associated Pareto gap; neither the revealed-preference test of Theorem~\ref{thm:MA_Af} nor the policy-gradient recursion guarantees of Theorem~\ref{thm:convg} depend on Nash structure. The mechanism-design conclusion, that a vanishing Pareto gap certifies socially optimal \emph{non-cooperative} play, requires, in addition, that those strategies constitute equilibria of the game induced by \(\phi\). Mixed-strategy Nash equilibrium is therefore adopted as an \emph{operational assumption}: it is the classical setting for the design problem, it renders the objective well-posed by guaranteeing (under Assumptions~\ref{as:pay}--\ref{as:pspace}) the existence of a mechanism that admits a socially optimal mixed-strategy Nash equilibrium, and the exposition is developed in that language. If the observations are instead generated by another behavioral protocol, the same procedure remains well-defined and continues to minimize the Pareto gap of the observed play. What is then obtained is a welfare improvement of that protocol, rather than implementation of social optimality at equilibrium. In all cases the mechanism is fitted to the particular strategies present in the data; we do not attempt to control unobserved alternative equilibria (e.g., in the sense of the price of anarchy).\footnote{The price of anarchy measures the discrepancy between socially optimal welfare and that of the \emph{worst-case} Nash equilibrium. Controlling this quantity is a strictly stronger requirement than the one considered here: we optimize the mechanism for the strategies that are observed, and have no access to other (in particular, worst-case) equilibria of the game.}}

\subsection{Context and Literature} Mechanism design arises in electronic market design \cite{yi2016fundamentals}, economic policy delegation \cite{mookherjee2006decentralization} and dynamic spectrum sharing \cite{sengupta2009economic}. \textcolor{black}{The approach proposed in this paper allows one to accomplish the mechanism design problem when only agents’ \emph{actions} are observed, while utilities remain latent or unreported. This setting includes bidding behavior in ad auctions
\cite{varian2007position,edelman2007internet},
power or channel selection in wireless networks
\cite{niyato2008dynamics},
and distributed coordination among UAVs or sensors
\cite{marden2009cooperative}.
 Here we briefly outline the standard approaches and contrast with our data-driven framework. }

\subsubsection{Mechanism Design Solution Concepts}
\label{sec:mdsc} Traditional approaches construct analytical transformations of the agent utility functions in order to induce \textit{social optimality} or \textit{truth-telling}. For instance, the classical Vickrey-Clarke-Groves (VCG) mechanism \cite{vickrey1961counterspeculation}, \cite{clarke1971multipart}, manipulates agent utilities by providing side payments to each agent which are proportional to the utilities gained by all other agents. This induces social optimality by making the utilities gained by \textit{all} agents of common interest to \textit{each} agent. While the VCG mechanism, and other similar analytical utility transformations \cite{d1979incentives},\cite{arrow1979property},\cite{myerson1981optimal}, have an elegant interpretation and straightforward mathematical structure, they require strict assumptions on the system. For instance in VCG it is assumed that the utility \textcolor{black}{valuations} of each agent are \textcolor{black}{reported} and that arbitrary side-payments are feasible. 

\textcolor{black}{A related class of mechanisms avoids direct utility elicitation through iterative, price-based interaction: ascending/clock auctions \cite{ausubel2002ascending, ausubel2004efficient} and Walrasian t\^atonnement \cite{arrow1958stability} each query agents' demand correspondence (their utility-maximizing bundle) at a sequence of posted prices, without requiring agents to report valuations directly \cite{nisan2006communication}. These approaches nonetheless presuppose a \emph{specific} mechanism form and a decoupled, typically quasi-linear optimization for each agent independently; the inter-agent coupling enters only through the designer's price-update rule. Our framework instead observes general, fully coupled joint mixed-strategies of the game under an arbitrary mechanism parametrization.}

Other approaches extend the frameworks of such analytical approaches, by providing automated or algorithmic methods to achieving mechanism design. The field of algorithmic mechanism design \cite{nisan1999algorithmic} \cite{conitzer2002complexity} aims to study the computational complexity of mechanism design approaches, arising out of the motivation to impose computational constraints on analytical solutions. Automated \cite{sandholm2003automated}, \cite{conitzer2004self}, \cite{shen2018automated} approaches aim to create new computational methods for constructing tailored mechanisms for the specific setting at hand.

Our framework lies within the area of \textit{empirical game theoretic analysis} (EGTA), which involves the study of games from specifications of only the strategy spaces and observed equilibria behavior. In particular, no a-priori knowledge \textcolor{black}{or reports} of the game structure are available, but one may simulate or observe samples from the game outputs. These frameworks are well-motivated by the necessity of data-driven game-theoretic settings, and are also referred to as simulation-based \cite{vorobeychik2008stochastic} or black-box \cite{picheny2019bayesian} game theory. Representative works which motivate the need for this data-driven setting occur in energy markets \cite{ketter2013autonomous}, supply chain management \cite{jordan2007empirical}, \cite{vorobeychik2006empirical}, network routing protocols \cite{wellman2013analyzing}, among others. \cite{viqueira2020empirical} presents a mechanism design approach for this data-driven EGTA framework, however they necessarily assume that the \textit{utilities} attained by sampled equilibria points are observable. We employ revealed preference theory within an IRL procedure to detect utility maximization behavior from observations of the agent strategies, then optimize the mechanism by an RL policy gradient algorithm. 

\textcolor{black}{The Bayesian approach to mechanism design assumes that both the designer and agents share access to a common prior distribution over utility functions or agent types. Under this assumption, the Bayesian Nash equilibrium (BNE) provides a natural and tractable solution concept. In  contrast, the meta-assumption in the current paper advances a fundamentally different modeling paradigm, namely a data-driven framework without access to type priors. }


\subsubsection{IRL via Revealed Preferences} Revealed preferences dates back to the seminal work \cite{samuelson1938note}, and establishes conditions under which it is possible to detect utility maximization behavior from empirical consumer data.  Afriat \cite{afriat1967construction} initialized the study of nonparametric utility estimation from \textcolor{black}{microeconomic} consumer budget-expenditure data, providing necessary and sufficient conditions under which such data is consistent with linearly-constrained utility maximization. \cite{forges2009afriat} extended this work to incorporate non-linear budget constraints. Existing works considering stochastic revealed preferences assume structural simplifications such as e.g., linear budget constraints \cite{bandyopadhyay1999stochastic}, finite action spaces \cite{mcfadden2006revealed}, or deterministic optimization with stochastic measurement errors \cite{aguiar2021stochastic}. 

Statement $i)$ of the meta-Theorem generalizes existing results in revealed preferences, \textcolor{black}{by providing necessary and sufficient conditions for a dataset of nonlinearly-constrained mixed strategies to be socially optimal, in the non-parametric utility framework}. 

\subsubsection{RL for Mechanism Design}
{\color{black} Recent studies have applied reinforcement learning (RL) techniques to automate mechanism design in dynamic and partially observed environments. For instance, \cite{shen2020reinforcement} propose \emph{Reinforcement Mechanism Design} (RMD), which learns dynamic pricing mechanisms in sponsored search auctions through online RL, assuming access to realized reward feedback for each allocation. \cite{lyu2022pessimism} extend this idea to the offline RL setting, combining pessimistic value estimation with VCG-style payment rules to learn near-optimal dynamic mechanisms from logged data. 

Beyond these RL formulations, several works have explored Bayesian and sequential learning methods for mechanism optimization. \cite{brero2019fast} introduce a Bayesian learning approach for iterative combinatorial auctions, where the mechanism designer updates a probabilistic model of agent valuations to accelerate convergence toward welfare-optimal allocations. \cite{brero2021reinforcement} develop an RL framework for \emph{sequential price mechanisms}, enabling the designer to learn adaptive price trajectories that approximate incentive-compatible outcomes in complex multi-stage environments.  \cite{brero2022stackelberg} model the design problem as a \emph{Stackelberg POMDP}, allowing the designer to optimize over policies that maximize expected social welfare under uncertainty about agents’ private information. As is clear from the meta-assumption on page~\pageref{massumption}, our framework substantially differs from these approaches.}

\subsubsection*{Organization} Section~\ref{sec:prelim} discusses background on mechanism design and revealed preference theory. Section~\ref{sec:pre:marp} provides necessary and sufficient conditions for a dataset of mixed-strategies to be consistent with social optimality (Theorem~\ref{thm:MA_Af}). In Section~\ref{sec:amd} we exploit Theorem~\ref{thm:MA_Af} to construct a loss function (w.r.t. mechanism parameters) which is globally minimized by mechanism parameters inducing socially optimal mixed strategy Nash equilibria. Here we also construct an SPSA  algorithm operating on this loss function using empirical mixed-strategy samples, and provide Theorem~\ref{thm:convg}, which states it converges in probability to the loss function's global minimizers. In Section~\ref{sec:robustrp} we provide an equivalence between non-globally-optimal loss function (Pareto gap) values and well-established relaxed revealed preference metrics, forming a bridge between our framework and these principles. In Section~\ref{sec:psdro} we develop a distributionally robust optimization procedure which achieves \textcolor{black}{the mechanism design objective} from only partial specifications (i.i.d. samples) of the mixed-strategies. 

\section{Mechanism Design. Background and Preliminaries.}

\label{sec:prelim}

This section outlines the notation and mathematical framework of the mechanism design problem that we consider. We first introduce the \textit{pure-strategy} framework, then extend to a more general \textit{mixed-strategy} framework which will form the basis of our subsequent investigation. 

\subsection{Pure-Strategy Game Structure}
\label{sec:psgs}
The pure-strategy mechanism design framework supposes a parametrized finite-player static non-cooperative pure-strategy game, structured by the tuple \[\textcolor{black}{G_{\text{pure}} = (\na,\CA,\mechspace,O,\{f^i\}_{i\in[\na]})}\]
where 
\begin{itemize}
    \item[-] $\na$ is the number of participating agents.
    \item[-] $\CA$ is the joint action space. Specifically, each agent $i \in \na$ may take an action (pure-strategy) $a^i \in \CA^i \subseteq \reals^{k}$, where $\CA^i$ is the space of actions available to agent $i$. Letting $\ba = [a^1,\dots,a^{\na}]'$ denote the \textit{joint-action} taken by all agents, $\CA \subseteq \reals^{k\na}$ is then the space of joint actions, given by the cartesian product $\CA = \otimes_i \CA^i$. 
    \item[-] $\mechspace$ is the mechanism parameter space, and $O$ is the outcome space. Specifically, a \textit{mechanism} $o_{\mech}: \CA \to O$, parametrized by $\mech \in \mechspace$, is a mapping from joint action $\ba$ to \textit{outcome} $o_{\mech}(\ba) \in O$.
    \item[-] $f^i: O \to \reals$ is the utility function of agent $i$, assigning to each outcome a real-valued utility.
\end{itemize}
The real-valued utility gained for agent $i$, given mechanism $o_{\mech}$ (with parameter $\mech$) and joint action $\ba$, is $f^i(o_{\mech}(\ba))$. For notational simplicity we denote 
\begin{equation}
\label{eq:fmap}
    f^i_{\mech}(\ba) = f^i(o_{\mech}(\ba))\textcolor{black}{: \reals^{k\na} \to \reals}
\end{equation} so that the emphasis is on the mapping from joint action $\ba \textcolor{black}{\in \reals^{k\na}}$ to utility $f^i$, with the mechanism $o_{\mech}$ acting as an intermediary such that $\mech$ parametrizes this mapping.

\begin{figure}[H]
\centering
\begin{tikzpicture}[node distance=2cm, >=Stealth]

    \node (A) {$\ba \in \CA$};
    \node[right=2.5cm of A] (B) {$o_{\mech}(\ba)\in O$};
    \node[right=2.5cm of B] (C) {$\{f^i_{\mech}(\ba)\}_{i=1}^{\na}$};

    \draw[->] (A) -- (B) node[midway, above] {mechanism $\mech$};
    \draw[->] (B) -- (C) node[midway, above] {utility};

\end{tikzpicture}
\caption{\small Pure-strategy mechanism flowchart. The mechanism modulates the mapping from joint-action to outcome space. The agent utilities are real-valued functions of outcomes. Thus, the mechanism equivalently modulates the utility functions' dependence on joint-actions.}
\label{fig:mechflow}
\end{figure}

Figure~\ref{fig:mechflow} illustrates the information flow in this setup. The mechanism is defined as a parameter which modulates the mapping from joint-action space $\CA$ to outcome space $O$; it thus modulates each agent $i$'s utility gained from joint-actions $\ba\in\CA$. \footnote{It is defined in this way, rather than directly parametrizing the utility functions, because oftentimes in practice the mapping from joint-actions to outcomes can be completely specified, while the influence on utility functions may be unknown.} 

\textit{Remarks}:
\begin{enumerate}[label=\roman*)]

\item \textit{Mechanism Parameter}: We distinguish the mechanism parameter $\mech$ from the mechanism $o_{\mech}$. The mechanism  $o_{\mech}$ is the \textit{map} from joint-actions $\ba$ to outcomes $o_{\mech}(\ba)$, while $\mech$ parametrizes this mapping. For instance, in the previous example we may have $\mech \in \{1,2\}$, $o_1 = \{\textrm{popular vote}\},\\ o_2 = \{\textrm{hierarchical electorate system}\}$, and $o_1(\ba) =$ \{candidate with the most votes in $\ba$ elected\}, $o_2(\ba) =$ \{candidate with most electoral votes based on $\ba$ elected\}.

\item \textit{Mechanism Design}: The mechanism design problem begins by assuming that each agent $i$ acts in its own self-interest, aiming to maximize its utility function $f_{\mech}^i$, by choosing a \textit{single} action (\textit{pure-strategy}) $a^i \in \CA^i$. Then, standard non-cooperative equilibrium concepts such as pure-strategy Nash equilibria can be invoked to formalize stability of joint-actions. The mechanism design problem is then to choose a mechanism (through parameter $\mech$) such that non-cooperative (Nash) equilibrium behavior also satisfies certain \textit{global optimality} conditions. 
\end{enumerate}
We expand on this design framework in the generalized context of  \textit{mixed-strategies}, where agents  choose a probability distribution over actions in order to maximize an outcome statistic such as \textit{expected} utility. 

\subsection{Mixed-Strategy Game Structure}
\label{sec:ms}
Since this paper focuses on RL for achieving mechanism design with mixed strategies, this section defines the mixed-strategy Nash equilibrium as a product measure. The mixed-strategy game form supposes a parametrized finite-player static non-cooperative mixed-strategy game, structured by the tuple 
\begin{equation}
    \label{eq:gmixed}
   \textcolor{black}{ G_{\text{mixed}} = (\na, \CA,\Delta(\CA), \mechspace,O, \{f^i\}_{i\in[\na]})}
\end{equation}
where $\na$ is the number of agents, $\{f^i\}_{i\in[\na]}$ are utility functions, $\CA$ is the joint action space, $\mechspace$ is the mechanism parameter space, and $O$ is the outcome space. Let $(\Omega,\mathcal{F})$ be standard measurable Euclidean space such that $\Omega = \reals^k$ and $\mathcal{F}$ is the Borel $\sigma$-field on $\reals^k$. Let $\mathcal{P}$ denote the space of probability measures on $(\Omega,\mathcal{F})$. Each agent $i\in\na$ can choose a probability measure (\textit{mixed-strategy}) $\msi \in \Delta(\CA^i)$, where $\Delta(\CA^i) \subseteq \mathcal{P}$ is the simplex of probability measures over $\CA^i \subseteq \Omega$. Then the \textit{joint mixed-strategy} is formed as the product measure\footnote{Since we work with mixed-strategy Nash equilibria in this paper, the product measure is a suitable construction. We discuss extension to other joint-strategies with inter-agent dependencies, such as correlated equilibria, in Section~\ref{sec:msmd}.}
\begin{equation}
\label{eq:prodmeas}
    \msb := \otimes_i \msi \in \Delta(\CA) := \otimes_i \Delta(\CA^i)
\end{equation}
where $\otimes_i$ denotes the product measure when applied to $\{\msi\}_i$, and the Cartesian product when applied to the simplices $\{\Delta(\CA^i)\}_i$, consistent with the notation of Section~\ref{sec:psgs}. Now, the real-valued \textit{expected} utility gained for agent $i$, given mechanism $o_{\mech}$ (with parameter $\mech$) and product measure $\msb$\footnote{We assume this measure admits a density w.r.t. the Lebesgue measure.} is 
\[f_{\mech}^i(\msb) := \CE_{\ba \sim \msb}\left[f^i(o_{\mech}(\ba))\right] = \int_{\CA}f^i(o_{\mech}(\ba))\msb(\ba)d\ba\]
We use the shorthand notation $f_{\mech}^i(\msb)$ to emphasize the mapping from product measure $\msb$ to expected utility, with mechanism $o_{\mech}$ acting as an intermediary such that $\mech$ parametrizes this mapping. Figure~\ref{fig:mechflowmixed} illustrates the information flow in this generalized setup. 

\begin{figure}[h]
\centering
\begin{tikzpicture}[node distance=2cm, >=Stealth]

    \node (A) {$\msb \in \Delta(\CA)$};
    \node[right=2.5cm of A] (B) {$o_{\mech}(\ba)\in O$};
    \node[right=2.5cm of B] (C) {$\{f^i_{\mech}(\ba)\}_{i=1}^{\na}$};

    \node[below=0.5cm of C] (D) {$\big\{\CE_{\ba\sim\msb}[f_{\mech}^i(\ba)] = \int_{\CA}f_{\mech}^i(\ba)\msb(\ba)d\ba\big\}_{i=1}^{\na}$};

    \draw[->] (A) -- (B) node[midway, above] {mechanism $\mech$} node[midway,below] {$\ba \sim \msb$};
    \draw[->] (B) -- (C) node[midway, above] {utility};

    \draw[->] (A) |- (D) node[near end, above] {expected utility under $\mech$};


\end{tikzpicture}
\caption{\small Mixed-strategy mechanism flowchart. Mechanism parameter $\mech$ modulates the mapping from joint-action to outcome space. Thus, $\mech$ equivalently modulates the utility functions' dependence on joint-actions. In the mixed-strategy regime each agent specifies a distribution $\msi$ over actions $a^i$, so that joint-actions $\ba$ are taken randomly from product measure $\msb$. The diagram splits two perspectives on this process: the top flowchart corresponds to the pure-strategy utility gained for \textit{a particular} joint-action $\ba$ taken randomly from $\msb$. However in the mixed-strategy regime agents aim to maximize their \textit{expected} utility over the product measure $\msb$, as represented by the bottom arrow.}
\label{fig:mechflowmixed}
\end{figure}

The mechanism design problem begins by assuming that each agent $i$ acts in its own self-interest, by choosing a mixed-strategy $\msi$ in order to maximize its expected utility $\CE_{\msb}[f_{\mech}^i(\cdot)]$. Observe that each agent's expected utility $\CE_{\msb}[f_{\mech}^i(\cdot)]$ depends on product measure $\msb$, and thus on every other agent's mixed-strategy. The standard solution concept for this non-cooperative interaction is the mixed-strategy Nash equilibrium:
\begin{definition}[Mixed-Strategy Nash Equilibrium]
   A product measure $\msb$ (defined in ~\eqref{eq:prodmeas}) is a mixed-strategy Nash equilibrium (MSNE) of the game $G_{\text{mixed}}$ under mechanism $o_{\mech}$ if \footnote{We assume the function $f_{\mech}^i(\cdot,\mu^{-i})$ is continuous, and since the simplex $\Delta(\CA^i)$ is compact, therefore the maximum exists. Detailed assumptions are provided in Section~\ref{sec:rpfa}.}
    \begin{equation}
        \label{eq:NE}
        \msi \in \arg\max_{\ms \in\Delta(\CA^i)}f_{\mech}^i(\ms,\ms^{-i}) = \arg\max_{\ms \in \Delta(\CA^i)}\CE_{\ba \sim (\ms,\ms^{-i})}\left[f^i_{\mech}(\ba)\right] \quad \forall i\in[\na]
    \end{equation}
    where $\ms^{-i} = \ms\backslash \msi$ is the product measure $\otimes_{j\neq i}\ms^j$, so that $(\msi,\ms^{-i}) = \msb$, and $f^i_{\mech}(\ms,\ms^{-i})$ is shorthand for the expected utility $f^i_{\mech}(\cdot)$ under product measure $(\ms,\ms^{-i})$. 
\end{definition}

Observe that the Nash equilibrium mixed-strategies \eqref{eq:NE} are dependent on the mechanism parameter $\mech$, \textcolor{black}{and Assumptions~\ref{as:confun} - \ref{as:pspace} will be sufficient to ensure the existence of a MSNE for every $\mech\in\mechspace$}. The \textit{mechanism design} problem is that of designing the mechanism $o_{\mech}(\cdot)$, through parameter $\mech$, such that \textcolor{black}{the \textit{resultant} MSNE $\mu(\mech)$} satisfies certain \textit{global welfare conditions}. 



Next we motivate the mechanism design objective of social optimality which we consider in the remainder of this work.

\subsection{Mixed Strategy Mechanism Design} 
\label{sec:msmd}
This section discusses the important concept of social optimality in mechanism design, which will form the objective of our RL approach. The goal of mechanism design is broadly to choose the mechanism $o_{\mech}$ such that Nash equilibria solutions are simultaneously \textit{globally optimal}, i.e., benefit the entire group. This section motivates the definition of 'globally optimal' through the social optimality criteria. 

\subsubsection{Design Objective. Social Optimality}The standard global welfare condition, also known as the utilitarian condition \cite{nisan1999algorithmic}, is to maximize the 'social welfare' 
\begin{equation}
\label{eq:soc_util}
\sum_{i=1}^{\na} f^i_{\mech}(\ms)
\end{equation} which is simply the sum of agent utility functions. This is referred to as social optimality.
\begin{definition}[Mixed-Strategy Social Optimality]
    A product measure $\msb$ is socially optimal for the game $G_{\text{mixed}}$ and mechanism $o_{\mech}$ if it maximizes the sum of agent utility expectations:
    \begin{equation}  
    \label{eq:sodef}
        \msb \in \arg\max_{\ms\in\Delta(\CA)}\sum_{i=1}^{\na} f^i_{\mech}(\ms) = \arg\max_{\ms\in\Delta(\CA)} \sum_{i=1}^{\na} \CE_{\ba \sim \ms}\left[f^i_{\mech}(\ba)\right]
    \end{equation}
\end{definition}

Notice that social optimality \eqref{eq:sodef} is a special case of Pareto-optimality, which maximizes the convex combination $\sum_{i=1}^{\na}\lambda_i f_{\mech}^i(\mu)$, $\sum_{i=1}^{\na}\lambda_i = 1, \lambda_i \geq 0$. In fact, the results presented in this work can easily be extended to the case where general Pareto-optimal solutions are considered. We exclude presentation and discussion of this generalization for brevity.

\emph{Remarks.}
\begin{itemize}
\item[-] \emph{Rank optimality.} In Appendix~\ref{sec:rankopt} we discuss  "rank optimality" in which the game outcome maximizes its rank among the agent's ordinal preference relations. This ordinal preference setting is a useful restriction \cite{chakrabarty2014welfare}, \cite{carroll2018mechanisms} that arises naturally from human preference relations. We prove that our criteria of social optimality implies the satisfaction of rank optimality. Thus, approaching mechanism design with the social optimality objective \eqref{eq:sodef} automatically guarantees that rank optimality will be met in restricted game conditions, enhancing the "optimality" of social optimality beyond aggregate utility maximization. In what follows we focus solely on social optimality as the mechanism design objective since it is more general.
\item[-] \emph{Strategic participant generality.} Our formulation is intentionally general: all strategic participants are represented within the common agent set $\na$. In auction environments this means that the seller (or auctioneer) is modeled as an additional agent whose utility may encode expected revenue. This differs from much of the auction-theoretic literature, where the seller defines the mechanism but is not treated as a strategic participant. Our representation enables a treatment of mechanism design in which any entity influencing the outcome—seller or buyer—may be incorporated symmetrically within the same multi-agent social welfare metric.
\item[-] \emph{Extension to correlated equilibria.}
While our current formalism focuses on independent mixed strategies, it can be extended to the more general notion of correlated equilibrium~\cite{aumann1974subjectivity}.
In that case, the agents would play from a joint distribution
$\boldsymbol{\mu}(a^1,\dots,a^M)$ satisfying joint incentive–compatibility conditions
$\CE_{\boldsymbol{\mu}}[f^i(a^i,a^{-i})] \ge \CE_{\boldsymbol{\mu}}[f^i(a_i',a^{-i})]$ for all~$i$ and deviations~$a_i'$. 
Our approach would remain applicable by simply determining a mechanism inducing social optimality
for the observed mixed strategies distributed according to this joint correlated equilibria.
\end{itemize}

\subsubsection{Mechanism Design Objective}

The mixed-strategy mechanism design problem is to find a mechanism $o_{\mech}$, through parameter $\mech$, such that the Nash equilibrium product measure $\msb$ arising from the conditions \eqref{eq:NE} simultaneously maximizes the social \textcolor{black}{welfare}, i.e.,
\begin{definition}[Mechanism Design]
\label{def:mechdes} With $\mechspace \subseteq \reals^p$ denoting the mechanism parameter space, the mechanism design problem is to find $\mech\in\Phi$ such that 
\begin{align}
\begin{split}
\label{eq:md}
    \msi \in \arg\max_{\ms\in\Delta(\CA^i)}f_{\mech}^i(\ms,\ms^{-i}) \ \forall i\in[\na] 
    \, \, \Rightarrow \, \, \msb \in \arg\max_{\ms\in\Delta(\CA)}\sum_{i=1}^{\na} f^i_{\mech}(\ms)
\end{split}
\end{align}
\end{definition}
Conceptually, this definition provides a formalization of designing a system such that even when individual agents act in pure self-interest, they produce outputs which benefit the entire group, both in the sense of maximizing social \textcolor{black}{welfare} \eqref{eq:sodef} and inducing rank optimality \eqref{eq:rankopt}. 

The mechanism design problem has been well-studied \cite{nisan1999algorithmic}, \cite{borgers2015introduction}, \cite{conitzer2002complexity} in the case where the agent utility functions $\{f^i_{\mech}(\cdot)\}_{i=1}^{\na}$ are known or can be designed. We consider the case when the utility functions are \textit{a-priori unknown} to the designer, and utility values \textit{cannot be observed empirically}. In many practical scenarios, the designer may need to design the system without explicit knowledge of the preferences and goals of its constituent members, and without access to the utility values resulting from observed strategies. We achieve this by utilizing the methodology of microeconomic \textit{revealed preferences} to formulate a stochastic optimization algorithm which converges to a mechanism parameter $\mech$ inducing the relation \eqref{eq:md}. 

\section{Identifying Socially Optimal Behavior: IRL via Revealed Preferences}
\label{sec:pre:marp}
This section presents our first main result (Theorem~\ref{thm:MA_Af}), which gives necessary and sufficient conditions for a dataset of observed nonlinear constraint sets and mixed-strategies to be consistent with social optimality \eqref{eq:sodef}. Theorem~\ref{thm:MA_Af} generalizes the results in \cite{cherchye2011revealed}, \cite{forges2009afriat}. \cite{forges2009afriat} extends Afriat's Theorem \cite{afriat1967construction} by providing necessary and sufficient conditions for a dataset of \textit{non-linearly} constrained responses to be consistent with utility maximization. \cite{cherchye2011revealed} provides conditions for a dataset of \textit{linearly} constrained \textit{multi-agent} responses are consistent with multi-objective optimization (Pareto-optimality).  as follows. We combine the two approaches to provide necessary and sufficient conditions for a dataset of \textit{non-linearly} constrained \textit{multi-agent} responses to be consistent with social optimality (multi-objective optimization). \textcolor{black}{This is a nontrivial extension and is of independent interest. This is the key construction that allows us to define our Pareto gap loss function, and thus to achieve the mechanism design setting in our nonparametric utility framework}.

\subsection{Revealed Preference Framework and Assumptions}
\label{sec:rpfa}
In order to state our first main result (Theorem~\ref{thm:MA_Af}), we specify the notation and structure of the observed data in a revealed preference setting. 
\subsubsection{Empirical Revealed Preference Framework}
Let us begin by assuming the same multi-agent mixed-strategy game form $G_{\text{mixed}}$ \eqref{eq:gmixed} as in Section~\ref{sec:ms}. We assume access to a dataset 
\begin{equation}
\label{eq:dataset}
    \CD = \{\CA_{\mech,t}, \msb_t \in \Delta(\CA_{\mech,t}), t\in[T]\}
\end{equation} of \textcolor{black}{constrained action spaces ("constraint sets")} $\CA_{\mech,t}$ and mixed-strategies \[\msb_t := \otimes_i \msi_t \in \Delta(\CA_{\mech,t})\] over a time period of $T$ steps. \textcolor{black}{Here we emphasize that the constraint sets themselves can be dependent on $\mech$, as part of the game-structure specified by the mechanism. }

\textcolor{black}{Each observation corresponds to an \emph{independent one–shot game}
played under a distinct constraint. 
Agents best–respond within that stage game, with respect to the time-varying constraint sets, but do not update strategies across time. }

For analytical convenience let us identify these constraint sets with structured "constraint functions". \textcolor{black}{These constraint functions will serve as the probe signals under which resultant agent equilibria behavior can be analyzed in a revealed preference framework.}

\textcolor{black}{\begin{definition}[Constraint Functions] 
\label{def:constraints} We suppose each constraint set $\CA_{\mech,t}^i$ is specified by a non-linear concave element-wise increasing constraint function $g_{\mech,t}^i: \reals^k \to \reals$ by 
\begin{equation}
\label{eq:CAti}
\CA_{\mech,t}^i = \{x\in\reals^k_+ : g_{\mech,t}^i(x)\leq 0\}, \,\,\, \CA_{\mech,t} = \otimes_{i=1}^{\na}\CA_{\mech,t}^i,
\end{equation}
and let $\Delta(\CA_{\mech,t}^i)$ denote the simplex of probability distributions over $\CA_{\mech,t}^i$. We quantify the constraint sets $\CA_{\mech,t}^i$ with constraint functions $g_{\mech,t}^i$ in order to make our setup amenable to revealed preference techniques, as will be made clear by the statement of Theorem~\ref{thm:MA_Af}.
\end{definition}}

In this work we take a data-driven approach to achieving mechanism design by \textit{only observing the dataset $\CD$}. We may thus re-frame the mechanism design objective in the following empirical consistency sense: 
\begin{definition}[Consistency with Social Optimality]
\label{def:emechdes}
Suppose we have access to a dataset $\CD = \{\msb_t,\CA_{\mech,t}\}_{t\in[T]}$, where $\msb_t$ are mixed-strategy Nash equilibria with respect to constraints $\CA_{\mech,t}$:
\[\msi_t(\mech) \in \arg\max_{\ms\in\Delta(\CA_{\mech,t}^i)}f_{\mech}^i(\ms,\ms_t^{-i}) \ \forall i\in[\na] \]
The empirical mechanism design problem is to find $\mech \in \mechspace$ such that \textbf{there exist} functions $f_{\mech}^i: \reals^{k\na} \to \reals, i\in[\na]$ \textcolor{black}{\eqref{eq:fmap}}, inducing  
\begin{equation}
\label{eq:rpso}
\msb_t \in \arg\max_{\ms\in\Delta(\CA_{\mech,t})}\sum_{i=1}^{\na}f_{\mech}^i(\ms) \,\, \forall t\in [T]
\end{equation}
If there exist functions $f_{\mech}^i$ satisfying social optimality \eqref{eq:rpso}, then we say the dataset $\CD$ is \textit{consistent with social optimality}. 
\end{definition}



Techniques for determining whether a dataset is consistent with utility maximization are well-studied in the field of microeconomic revealed preferences. Previous work in this field was initialized by the pioneering results of Afriat \cite{afriat1967construction}, who constructed a set of linear inequalities that are \textit{necessary and sufficient} for a dataset of (linear) constraints and (pure-strategy) responses to be consistent with (single-agent) utility maximization. This work was subsequently generalized by Forges \& Minelli \cite{forges2009afriat} to account for nonlinear budget constraints. In Theorem~\ref{thm:MA_Af} we extend these seminal works by providing necessary and sufficient conditions for the dataset $\CD$ to satisfy \eqref{eq:rpso}. This result will form the basis of our algorithms for automated mechanism design. 

\subsubsection{Generalized Axiom of Revealed Preference}  In revealed preference theory, the generalized axiom of revealed preference (GARP) is a key definition that serves as a necessary and sufficient condition for a dataset to be consistent with utility maximization. Here we introduce a mixed-strategy multi-agent generalized axiom of revealed preference (MM-GARP), analogous to GARP in \cite{forges2009afriat} for the pure-strategy, single-agent case. 

\begin{definition}[MM-GARP]
\label{def:MMGARP}
We say that strategy $\msi_k \in \Delta(\CA_k^i)$ is \textit{revealed preferred} to strategy $\msi_j \in \Delta(\CA_j^i)$, denoted $\msi_k \,R \,\msi_j$, if $g_k^i(\msi_j) := \int_{\CA_j^i}g_k^i(a)\msi_j(a)da \leq g_k^i(\msi_k) := \int_{\CA_k^i}g_k^i(a)\msi_k(a)da$. Let $H$ be the transitive closure of the relation $R$. The dataset $\CD$ satisfies the mixed-strategy generalized axiom of revealed preferences (MM-GARP) if, for any $i\in[\na]$, $k,j\in[T]$, $\msi_k \, R\, \msi_j$ implies $g_j^i(\msi_k) \geq g_k^i(\msi_k)$.
\end{definition}

\textcolor{black}{MM-GARP will be useful in characterizing the necessary and sufficient conditions for dataset $\CD$ to be consistent with social optimality \eqref{eq:sodef}, and in our presentation of relaxed revealed preference metrics in Section~\ref{sec:robustrp}. In particular, a relaxation of MM-GARP leads to an interpretable quantification of algorithmic suboptimality for our mechanism design approach, allowing us to reason about effectiveness of achieved non-socially-optimal mechanisms.}

\subsection{Modeling Considerations}
\subsubsection{Motivation for repeated one-shot dataset collection.}

{\color{black} Each data pair $(\CA_{\mech,t},\msb_t)$ corresponds to the action space and resultant equilibria strategy of an independent one-shot game. Our formulation should not be interpreted as a repeated–game setting. \textcolor{black}{Specifically, agents best-respond within each stage game without updating across stages. Any learning occurs either prior to the designer initialization, or \emph{within} the stage-game itself through an iterative learning dynamics. At best, then, the relevant information incorporated into agent play in each stage is a complete characterization of opponent's utilities, inducing a complete-information equilibrium within each stage; our framework readily handles this. 
Also, we emphasize that repeated-game equilibria frequently support socially-optimal outcomes \emph{automatically} through reputation and cooperation \cite{fudenberg1990evolution}, in which case these ``good" outcomes are attributable to the repeated structure rather than the mechanism. Thus, by targeting one-shot equilibria, our Pareto-gap objective isolates the \emph{mechanism's own} causal contribution to welfare. Furthermore, in our model formulation the inter-stage incentive structure of repeated games cannot be expressed since we \emph{isolate incentives (constrained utilities) within each distinct stage game}.}


This sequence of independent one-shot games \textcolor{black}{is a standard structural specification in adaptive/online mechanism design settings \cite{9209172}, repeated posted-price auctions with myopic buyers \cite{braverman2018selling}, or repeated Stackelberg security games \cite{balcan2015commitment}}. Furthermore, it models observational setups of pertinent multi–agent settings in electronic sensing and communication, e.g., radar–UAV interactions \cite{snow2022identifying} or resource–allocation networks, where
the environment issues a time–varying constraint $\CA_{\mech,t}$ (power, bandwidth, detectability, etc.)
while the agents' internal decision rules (leading to equilibria strategies) remain fixed. 
Each interaction can therefore be viewed as an independent realization of a stage game
under a new constraint. This abstraction allows the designer to infer or optimize
the underlying utilities or mechanism parameters from diverse but independent samples,
without needing to model long–horizon learning or adaptation. }

\subsubsection{Rationale for equilibrium achievement.}
\label{sec:eqrat}
{\color{black}  While Nash equilibria are well-motivated in several applications, constructing multi-agent interacting systems from which these equilibria emerge can be difficult. 
Classical works \cite{fudenberg1998theory,hofbauer2002global} and more recent work \cite{bichler2023convergence, bichler2025agentic} on learning in Bayesian auction games and on markets with learning agents characterizes conditions under which convergence does or does not occur. The framework of this paper is compatible with these learning dynamics, if they occur as intra-stage procedures which produce a fixed joint strategy in each stage-game, regardless of whether this joint strategy is an equilibrium. Recall the discussion of this applicability to non-equilibrium behavior in Section~\ref{sec:motddrl}.}

\subsection{Modeling Mechanism-Dependent Action Space Structure} 
\label{sec:confungen}
\textcolor{black}{In revealed preference theory, it is customary to have a dataset of endogenous inputs (constraints) and resultant outputs (equilibria distributions, in our case). This input-output dataset pairing is essential to the mathematical characterization of a dataset's consistency with optimality. In this section we provide an assumption on the structure of the imposed constraint functions which determine these actions spaces; this serves as a sufficient specification for "probing" the system which allows \textit{exact} determination of consistency with social optimality (through necessary \textit{and} sufficient conditions), derived in Section~\ref{sec:main1}. Furthermore, this modeling structure allows the constraint specifications to be dependent on the mechanism $\mech$, which is a common structural assumption in mechanism design.} 

We introduce the following assumption on the structure of constraint functions $\{g_{\mech,t}^i(\cdot), t\in[T]\}_{i=1}^{\na}$. 

\begin{assumption}[Constraint Function Structure]
\label{as:confun} 
    Assume that\\
    a) The constraint functions $\{g_{\mech,t}^i: \reals^{\color{black} k} \to \reals\}$ are each concave and increasing in each dimension.\\ 
    b) For each $i\in[\na]$ there is some vector $\beta^i \in \reals^k$ such that for any $a\in\reals, t\in[T]$, \[x,y \in \{\gamma \in \reals^{\color{black} k} : g_{\mech,t}^i(\gamma-a\beta^i) = 0\} \Rightarrow g_{\mech,t}^i(x) = g_{\mech,t}^i(y)\] i.e., the equivalence class of points in the zero level-set is invariant under shifts of $g_{\mech,t}^i(\cdot)$ by $\beta^i$.
\end{assumption}

\noindent \textit{Remark. Modeling }. Assumption~\ref{as:confun} will be necessary for the proof of our extended result in Theorem~\ref{thm:MA_Af}, and is easily satisfied when we, the analyst, have control over the sequential constraint functions $g_{\mech,t}^i$. If, so we may generate the constraint functions $g_{\mech,t}^i$ as follows:
\begin{enumerate}[label=\roman*)]
    \item Choose base functions $\{g_{\mech}^i: \reals^{\color{black} k} \to \reals\}$ satisfying: for each $i\in[\na]$ there is some vector $\beta^i \in \reals^k$ such that for any $a\in\reals, t\in[T]$, \[x,y \in \{\gamma : g_{\mech}^i(\gamma-a\beta^i) = 0\} \Rightarrow g_{\mech}^i(x) = g_{\mech}^i(y)\] 
    \item Generate $T$ random variables $\{z_t\}_{t=1}^T$ uniformly from a compact positive Lebesgue-measure set $\chi \subset \reals$, and form each constraint function as $g_{\mech,t}^i(\cdot) = g_{\mech}^i(\cdot - z_t\beta^i)$.
\end{enumerate}
\textit{Example}: Consider the simple linear constraint function $g^i(x) = \alpha' x -1 \leq 0$, where $\alpha, x \in \reals^{\color{black} k}$. Then taking $c\in\reals^{\color{black} k}$, for any $x,y$ such that $g^i(x-c) = \alpha' (x-c) - 1 = g^i(y-c) = \alpha' (y-c) - 1 = 0$, we have $g^i(x) = \alpha'x = \alpha'c - 1 = \alpha'y = g^i(y)$. We may also construct a nonlinear example as follows: $g^i(x) = \log(2\sigma(\sum_{i=1}^n x_i)) \leq 0$, where $\sigma$ is the sigmoid function $\sigma(x) = \frac{1}{1 + e^{-x}}$. Then, taking $c = [c_i,\dots,c_n]$, for any $x,y$ such that $g^i(x-c) = \log(2\sigma(\sum_{i=1}^n (x_i  -c_i))) = \log(2\sigma(\sum_{i=1}^n (y_i  -c_i))) = g^i(y-c) = 0$ we have  $\sum_{i=1}^n x_i = \sum_{i=1}^n c_i = \sum_{i=1}^n y_i$, and thus $g^i(x) = g^i(y)$.

 Control over the constraint functions is commonly assumed in many application of revealed preferences \cite{dong2018strategic}, \cite{snow2023statistical}, where the constraint function is treated as a "probe" that allows for interaction between the analyst and observed system. As we shall see, the mechanism design setting naturally motivates this control, since we assume a-priori that we as the analyst have some control over the system parameters (such that we can specify particular mechanisms to begin with). 

\paragraph{Sensitivity and Robustness.}
\textcolor{black}{Assumption~1 provides sufficient structure on the constraint functions to ensure revealed preference
testability and existence of a socially optimal mixed-strategy Nash equilibrium.
In practice, these constraint functions may be only approximately of this form or partially outside the
analyst’s control. The distributionally robust optimization (DRO) formulation introduced in
Section~\ref{sec:psdro} can directly accommodate this case.
By extending the ambiguity set from distributions over strategies to include bounded perturbations of
the constraint functions,
the resulting robust problem has the \emph{same structure and convergence properties}
as the DRO formulation in Theorem~\ref{thm:dreq}.
Hence, our existing robust mechanism design algorithm can handle imperfect or partially controlled
constraint functions with no change in implementation---only control of the ambiguity set.}

\subsection{Motivating Example: Mechanism Design under Information Asymmetry}
\label{sec:netex}
{\color{black}
Here we illustrate a setting in which the information asymmetry between agents and designer arises naturally,
while the assumption of equilibrium behavior remains well-posed. A detailed numerical example is provided for this setup, in Section~\ref{sec:num_wireless}. 

Consider a shared-spectrum market
comprising $M$ self-interested communication firms (agents) and a regulatory authority acting as
the mechanism designer. The regulator specifies a pricing mechanism $\mech$ determining how spectrum
access costs depend on congestion, while the firms select their transmission powers or bandwidth
demands in response.

\paragraph{Agent utilities.}
Each firm $i \in [M]$ seeks to maximize its \emph{private} utility
\begin{equation}
    f_{\mech}^i(a_i, a_{-i})
    \;=\;
    v_i \, q_i(a_i, a_{-i})
    \;-\;
    c_i(a_i; \mech),
    \label{eq:private-utility}
\end{equation}
where
$a_i \in A_i$ is the firm’s action (e.g., transmission power),
$q_i(a_i, a_{-i})$ is the achieved throughput depending on congestion,
$v_i>0$ is the firm’s private valuation of throughput,
and $c_i(a_i; \mech)$ is the cost induced by the designer’s pricing mechanism.
Such formulations are standard in network economics and congestion games
\citep{johari2004efficiency,roughgarden2005selfish}.

\paragraph{Information structure.}
The key asymmetry is as follows:
\begin{align*}
    \text{(Agents)} &:\;
        \text{Each agent knows its own $v_i$ and $f_{\mech}^i$, and can estimate others' impact via feedback.}\\ 
    \text{(Designer)} &:\;
        \text{Does \emph{not} observe any $v_i$ or $f_{\mech}^i$, but can query the system by supplying $\mech$}\\
        &\quad\text{and observing the induced equilibrium behavior $\mu(\mech)$.}
\end{align*}
This structure mirrors classical private-information models in mechanism design
\citep{myerson1981optimal},
and modern empirical mechanism design frameworks
\citep{vorobeychik2006empirical,viqueira2020empirical}.
The designer treats the multi-agent system as a black box that returns a joint mixed
strategy $\mu(\mech)\in \mathcal{N}(\mech)$, the mixed Nash equilibrium of the game induced by~$\mech$.

\paragraph{Mechanism design objective.}
The regulator aims to tune $\mech$ so that the observed equilibria $\mu(\mech)$ become socially
optimal, i.e., to find $\mech^*$ satisfying
\begin{equation}
    \mech^* \in
    \arg\min_{\mech\in\mechspace}
    L(\mu(\mech))
    \quad\text{where}\quad
    L(\mu(\mech)) =
    \Bigg|
    \sum_{i=1}^{\na} f_{\mech}^i(\mu(\mech))
    - \max_{\mu\in\Delta(A)} \sum_{i=1}^{\na} f_{\mech}^i(\mu)
    \Bigg|.
\end{equation}
Since the utilities are unobserved, the designer estimates
$L(\mu(\mech))$ indirectly via a revealed preference or inverse-reinforcement-learning step,
then updates $\mech$ by stochastic policy-gradient methods as developed in Sections~4–5
\citep{ng2000algorithms,krishnamurthy2020identifying}.

\paragraph{Real-world analogues.}
This asymmetric information structure is ubiquitous:
\begin{itemize}
    \item[-] In electricity and spectrum markets, operators observe aggregate equilibria but not private utilities
          \citep{mookherjee2006decentralization,sengupta2009economic}.
    \item[-] In online advertising and auction design, bidders’ valuations are hidden, yet bids correspond to
          equilibrium strategies \citep{edelman2007internet}.
    \item[-] In ridesharing or logistics platforms, dynamic pricing is tuned from behavioral data without access
          to individual utilities \citep{bimpikis2019spatial}.
\end{itemize}
Hence, the assumption that the designer observes only equilibrium strategies,
while agents possess local utility information, is both realistic and essential to the
data-driven mechanism design paradigm considered in this work.
} \textcolor{black}{As discussed in Section~\ref{sec:mdsc}, other techniques, e.g., iterative demand-query pricing mechanisms, do indeed work with this information asymmetry, but our technique is motivated both for its generality of mechanism  parametrization and its ex-post certification of social optimality achievability.}

\subsection{Main Result I. Dataset Consistency with Social Optimality}
\label{sec:main1} Now we present our first main result, providing necessary and sufficient conditions for a dataset $\CD = \{\CA_{\mech,t},\msb_t\}_{t\in[T]}$ to be consistent with social optimality \eqref{eq:rpso}. This result is of independent interest. It will also serve as the key tool allowing us to achieve automated mechanism design.

\begin{theorem}
\label{thm:MA_Af}
    Consider a dataset $\dataset = \{\CA_{\mech,t},\msb_t, t \in [T]\}$ \eqref{eq:dataset} comprising constraints and mixed strategies. Under Assumption~\ref{as:confun}, the following are equivalent:
    \begin{enumerate}[label=\roman*)]
        \item There exists a set of concave, continuous, locally non-satiated utility  functions $\{f_{\mech}^i: \reals^{k\na} \to \reals\}_{i=1}^\na$ such that the mixed-strategies $\{\msb_t, t\in [T]\}$ are socially optimal, namely \eqref{eq:rpso} holds.
        \item $\CD$ satisfies MM-GARP, given in Def.~\ref{def:MMGARP}.
        \item For each $i \in \textcolor{black}{[\na]}$, there exist constants $u_t^i\in\reals, \lambda_t^i > 0$, such that the following inequalities hold:
        \begin{equation}
        \label{eq:ineq}
            u_s^i - u_t^i - \lambda_t^i(\nlconst^i(\msb_s)) \leq 0 \quad \forall t,s \in [T]
        \end{equation}
        where $g_{\mech,t}^i(\msb_s) = \int_{\CA_s^i}g_{\mech,t}^i(a)\msi_s(a)da$
    \end{enumerate}
\end{theorem}

\begin{proof}
    See Appendix~\ref{ap:pfa}
\end{proof}

\textit{Remarks}:
\begin{enumerate}[label=\roman*)]
\item \textit{Concavity Non-Refutation}: \textcolor{black}{A remarkable feature of Afriat's Theorem \cite{afriat1967construction}, and of Theorem~\ref{thm:MA_Af}, is that concavity of utilities cannot be refuted by a finite dataset of agent behavior.}
\item \textit{Testing for Consistency with Social Optimality}: Theorem~\ref{thm:MA_Af} extends the  result of \cite{forges2009afriat} (which deals with a single agent with a non-linear budget constraint) to mixed-strategies and multiple agents. It provides \textit{testable} necessary and sufficient conditions for a dataset to be consistent with social optimality, i.e., for $i)$ to hold. Indeed, through the dataset $\CD$, the analyst has access to constraint functions $g_{\mech,t}^i$ and mixed-strategies $\msb_t$ for all $t$; so she must only test the feasibility of the linear program \eqref{eq:ineq}, with decision variables $u_t^i,\lambda_t^i$, to determine consistency. Furthermore, there are well-established methods for determining if GARP holds given an empirical dataset, see \cite{cherchye2011revealed}; MM-GARP is directly amenable to such methods and so testing condition $ii)$ is a feasible alternative to testing condition $iii)$ of the theorem.

\item \textit{Connection to Mechanism Design}: In Section~\ref{sec:prelim} we introduced the setting of mechanism design, where the goal is to find a system parameter that induces social optimality from non-cooperative Nash equilibria. Theorem~\ref{thm:MA_Af} provides a novel result in the setting of multi-agent revealed preferences, which gives necessary and sufficient conditions for a dataset of sequentially observed constraint sets and mixed-strategy responses to be consistent with social optimality. This is the key step in formulating the Pareto gap, which will be introduced in the following section.


\item \textit{Generalization by Mixed-Integer Linear Programming}: In many microeconomic situations it is natural to observe only a subset of agent actions, e.g., a subset of all budget-constrained consumption behavior. Such cases are handled in revealed preference theory by introducing "assignable variables", which are feasible if there exist extensions of the observed behavior which is consistent with utility maximization behavior. These scenarios can be handled by generalizing the inequalities \eqref{eq:ineq} to a mixed integer linear program, see \cite{cherchye2011revealed} or \cite{snow2022identifying}. This generalized setting is a natural extension of the framework developed in this paper.
\end{enumerate}

\section{RL Mechanism Design using Policy Gradient Optimization}
\label{sec:amd}
In this section we introduce an empirical framework where a designer can provide mechanisms and observe mixed-strategy Nash equilibria system responses. In this framework, we utilize Theorem~\ref{thm:MA_Af} to construct a computational method for achieving mechanism design. Specifically,
\begin{enumerate}[label=\roman*)]
\item In Section~\ref{sec:algder}, we exploit Theorem~\ref{thm:MA_Af} to derive a non-negative loss function that measures the \textit{proximity} of the observed mixed-strategy Nash equilibria to the social optimum \eqref{eq:sodef} (we denote this the "Pareto gap"). We construct a Pareto gap which operate on both fully-observed and partially-observed mixed strategies.

\item  Sections~\ref{sec:spsa} and \ref{sec:convg} present a novel algorithmic mechanism design procedure. \textcolor{black}{In Section~\ref{sec:spsa}, we construct a stochastic approximation algorithm, Algorithm~\ref{alg:amd}, to minimize the Pareto gap constructed in Section~\ref{sec:algder}. Section~\ref{sec:convg} provides our second main result, Theorem~\ref{thm:convg}, which states that Algorithm~\ref{alg:amd} converges in probability to the set of global minimizers of the Pareto gap. Algorithm~\ref{alg:amd} thus serves as a \textit{test for achievability of social optimality}: if the Pareto gap minima occurs at zero then social optimality is achievable and Algorithm~\ref{alg:amd} produces a mechanism achieving it; else, it is identified that social optimality is unachievable for the game setup at hand.} 

\end{enumerate}

\subsection{Modeling Framework and Construction of a Pareto gap Loss Function}
\label{sec:algder}
Here we describe the high-level strategy and motivation for our adaptive mechanism design algorithm. We suppose the mixed-strategy game setting of Section~\ref{sec:ms}, and adopt the revealed preference notation of Section~\ref{sec:pre:marp} such that we introduce an empirical observation window of $T$ time-steps. Thus we consider the game form
\begin{equation}
    \label{eq:Gmixed}
    G_{\text{mixed}} = \left(\na, \{f_{\phi}^i\}_{i\in[\na]},\{\CA_{\mech,t}\}_{t=1}^T,\{\Delta(\CA_{\mech,t})\}_{t=1}^T,\mechspace,O\right)
\end{equation}
where each constraint set $\CA_{\mech,t}$ is defined as in Section~\ref{sec:pre:marp}, with $\CA_{\mech,t} = \otimes_{i=1}^{\na}\CA_{\mech,t}^i$ and $\CA_{\mech,t}^i$ defined by the constraint function $g_{\mech,t}^i: \reals^k \to \reals$ \eqref{eq:CAti}. Recall $\Delta(\CA_{\mech,t})$ denotes the simplex of probability distributions over set $\CA_{\mech,t}$.

\subsubsection{Interaction Procedure} First we outline the interaction procedure between the designer and the system constituents. As in the revealed preference framework, we treat the multi-agent system as a black box. The designer can provide a mechanism $o_{\mech}$ and constraints $\{g_{\mech,t}^i(\cdot)\}_{i=1}^{\na}$. It subsequently observes a \textit{mixed-strategy Nash equilibrium} response $\msb_t(\mech)$ \eqref{eq:NE} of the group with mechanism parameter $\mech$ and simplex strategy constraint $\Delta(\CA_{\mech,t})$, where
\[\msb_t(\mech) :=\otimes_i \msi(\mech) \in \Delta(\CA_{\mech,t})\] We denote the probability that agent $i$ takes action $a$ under mixed strategy $\msi(\mech)$ as $\ms^i(a;\mech)$. We denote $\msb_t(\mech)$ as a function of $\mech$ to emphasize the dependence of Nash equilibria on the mechanism $o_{\mech}$. The assumption that the response $\msb_t(\mech)$ is a Nash equilibrium is standard in the mechanism design literature, and is readily achieved from non-cooperative best-response interactions among the agents \cite{fudenberg1991game}, under reasonable assumptions on the game structure. This is an operational assumption we make throughout to be consistent with the standard mechanism design concept, but is not technically necessary for our results to hold; see the discussion in Section~\ref{sec:motddrl}. Figure~\ref{fig:int_proc} displays this interaction procedure.

\begin{figure}[h]
    \centering
    \begin{tikzpicture}
        \node[draw, rectangle, minimum width=1cm, minimum height=1cm, align=center] (designer) at (0,0) {Designer};
        
        \node[draw, rectangle, minimum width=1cm, minimum height=1cm, align=center] (mas) at (4,0) {multi-agent system \\ $\{f^i_{\mech}(\cdot)\}_{i=1}^\na$};
        
        \draw[->] (designer.east) -- (mas.west) node[midway, above, font=\small] {$o_{\mech}$};

        \draw[->] (designer.east) -- (mas.west) node[midway, below, font=\small] {$\{g_{\mech,t}^i(\cdot)\}_{i=1}^\na$};

        \draw[->] (mas.south) -- +(0,-0.5)  -|
        (designer.south) node[near start,below,font=\small] {mixed-strategy Nash equilibrium $\msb_t(\mech)$};

    \end{tikzpicture}
    \caption{\small Interaction procedure. At each time $t$, the designer can provide mechanism $o_{\mech}$ and constraint functions $\{g_{\mech,t}^i(\cdot)\}_{i=1}^\na$, and observes mixed-strategy Nash equilibrium response $\msb_t(\mech)$.}
    \label{fig:int_proc}
\end{figure}

Suppose the designer now provides $T$ sets of constraint functions $\{g_{\mech,t}^i(\cdot), t\in[T]\}_{i=1}^\na$ and receives $T$ corresponding Nash equilibrium responses $\msb_t(\mech)$ under mechanism parameter $\mech$. We denote this $\mech$-dependent dataset as 
\begin{equation}
\label{eq:dset}
    \CD_{\mech} = \{\{g_{\mech,t}^i(\cdot)\}_{i=1}^{\na}, \msb_t(\mech), t\in[T]\}
\end{equation}
 We translate the mechanism design goal to this empirical setting: to find a mechanism parameter $\hat{\mech}$ such that empirical Nash equilibrium responses $\msb_t(\hat{\mech})$ are \textit{consistent} with social optimality, i.e., \eqref{eq:rpso} holds. We do this by constructing a Pareto gap loss function which we can then minimize via stochastic optimization.\\



\subsubsection{Formulating the Pareto gap.} 
\label{sec:parreg} Recall that the designer can directly test whether the dataset $\CD_{\mech}$ is consistent with social optimality by testing the feasibility of linear program \eqref{eq:ineq}. That is, if \eqref{eq:ineq} is feasible, then the goal is achieved. Our adaptive mechanism design solution is then to iteratively modify the parameter $\mech$ until \eqref{eq:ineq} is feasible for dataset $\CD_{\mech}$. We will now formulate a loss function $L(\mech)$, quantifying the Pareto gap, which is minimized at some $\hat{\mech}$ inducing feasibility of \eqref{eq:ineq}.

 Suppose the LP \eqref{eq:ineq} does not have a feasible solution for the dataset $\mathcal{D}_{\mech}$. How can we quantify the proximity of $\mech$ to a mechanism $\hat{\mech}$ which induces feasibility of \eqref{eq:ineq}? Notice that we can augment the linear program \eqref{eq:ineq} as:
\begin{align}
\begin{split}
\label{eq:ineq_aug}
    &L(\mech) = \min_{r\geq 0} : \exists \{u_t^i \in \reals,\lambda_t^i > 0, t\in[T], i\in[\na]\} :\\ &\quad u_s^i - u_t^i - \lambda_t^i\nlconst^i(\msb_s(\mech)) \leq \lambda_t^i r \quad \forall t,s\in[T],i\in[\na] \\
    & \quad \textrm{ where }\nlconst^i(\msb_s(\mech)) = \int_{\CA_{\mech,t}}g_{\mech,t}^i(a)\ms_s^i(a;\mech)da
\end{split}
\end{align}
where $\ms_s^i(\cdot\,;\mech)$ is the mixed-strategy distribution of agent $i$ at time $s$, under mechanism parameter $\mech$. 

\textcolor{black}{\textit{Remark. Notation.} Here, for notational convenience, we suppress the dependency of $L(\mech)$ on the full dataset $\dataset_{\mech}$ containing specific equilibria $\{\mu_t(\mech)\}_{t=1}^T$, and denote its dependency only on the mechanism $\mech$.} 

Here $L(\mech)$ is the minimum value of $r$ such that the specified linear program has a feasible solution for the dataset of mixed-strategy Nash equilibrium responses $\{\msb_s(\mech),s\in[T]\}$. 
Thus, $L(\mech)$ can be treated as a metric\footnote{$L(\mech)$ is not a true metric in the precise mathematical sense. It is a measure of proximity of the dataset $\CD_{\mech}$ to optimality. We prove several formal results about its structure that make it a useful "metric", e.g., Theorems \ref{thm:cineq}, \ref{thm:robusteq}, \ref{thm:dreq}.} capturing the proximity of the dataset $\mathcal{D}_{\mech}$ feasibility \eqref{eq:ineq}: if $L(\mech) = 0$, then \eqref{eq:ineq} is feasible, else lower $L(\mech)$ indicates $\CD_{\mech}$ being "closer" to feasible.

\begin{lemma}
\textcolor{black}{The following structural properties on $L(\mech)$ hold.}
\begin{itemize}
    \item[-] \textcolor{black}{$L(\mech)$ is thrice continuously differentiable: $L(\mech) \in C^3$.}
    \item[-]  $L(\mech) = 0$ if and only if the dataset $\CD_{\mech}$ \eqref{eq:dset} is consistent with social optimality: \eqref{eq:rpso} holds.
\end{itemize}
\end{lemma}
\begin{proof}
The differentiability is non-trivial, and is proved as H3 in the proof of Theorem~\ref{sec:con_pf} in Appendix~\ref{sec:proofs}. The second property follows immediately from Theorem~\ref{thm:MA_Af}.
\end{proof}

Thus, we can adopt the augmented linear program solution $L(\mech)$ as a \textcolor{black}{smooth} \textit{objective function} taking parameter $\mech$ and outputting a value representing the proximity of the empirical dataset $\CD_{\mech}$ to feasibility (social optimality). The adaptive mechanism design problem can then be re-stated as 
\begin{equation}
    \label{eq:conmin}
    \textrm{Find } \hat{\mech}\in\mechspace \, \,s.t.\,\, \hat{\mech}\in\arg\min_{\mech} L(\mech)
\end{equation}
In Section~\ref{sec:convg} we will introduce structural assumptions such that $\min_{\mech} L(\mech) \leq 0$, making mechanism design feasible; such that $\CD_{\arg\min_{\mech}L(\mech)}$ satisfies $iii)$ of Theorem~\ref{thm:MA_Af} and thus \eqref{eq:rpso}.

\subsubsection{Formulating a Stochastic Loss Function from Partially Observed Strategies}

Here we consider the case when at each time step $t\in[T]$, the full density of mixed-strategy $\msb_t(\mech)$ is not observed directly. Instead, we assume that only $N \in \nat$ i.i.d. samples $\{\gamma_{t}^k(\mech)\}_{k=1}^N$ from $\msb_t(\mech)$ are observed. This arises in practice from e.g., the dynamics of a finite repeated game \cite{benoit1984finitely}. We then consider the empirical distribution 
\begin{equation}
\label{eq:empdist}
    \hat{\msb}_t(\cdot\,;\mech) := \frac{1}{N}\sum_{k=1}^N \delta(\cdot - \gamma_{t}^k(\mech))
\end{equation}
where $\delta(\cdot)$ is the standard Dirac delta.

Consider the following stochastic Pareto gap loss function, constructed as the solution of a linear programming problem whose parameters are specified by this random empirical distribution:
\begin{align}
\begin{split}
\label{eq:ineq_aug2}
    &\hat{L}(\mech) = \min_{r\geq 0} : \exists \{u_t^i \in \reals,\lambda_t^i > \alpha, t\in[T], i\in[\na]\} :\\ &\quad u_s^i - u_t^i - \lambda_t^i\nlconst^i(\hat{\msb}_s(\mech)) \leq \lambda_t^i r \quad \forall t,s\in[T],i\in[\na] \\
    & \quad \textrm{ where }\,g_{\mech,t}^i(\hat{\msb}_s(\mech)) = \frac{1}{N}\sum_{k=1}^N g_{\mech,t}^i(\gamma_{s}^k(\mech))
\end{split}
\end{align}
Then, we have the following result relating $\CE[\hat{L}(\mech)]$ to $L(\mech)$:

\begin{proposition} 
\label{thm:cineq}
Fix a $\mech \in \mechspace$ and $c \geq 0$, and suppose $\CE[\hat{L}(\mech)] = c$.  Then for any $\epsilon\geq 0$,

\[\mathbb{P}\left(L(\mech) \leq c + \epsilon \right) \geq \prod_{t=1}^T \prod_{i=1}^{\na} \left(\max\biggl\{1-2\exp\left(\frac{-2\epsilon^2N}{ G^2 }\right), 0 \biggr\}\right)^T\]

where $G := |\inf_{i,t} g_{\mech,t}^i(\boldsymbol{0})|$, $L(\mech)$ is given in \eqref{eq:ineq_aug} and $\hat{L}(\mech)$ is given in \eqref{eq:ineq_aug2}. 
\end{proposition}
\begin{proof}
See Appendix~\ref{sec:lempf}.
\end{proof}

Proposition~\ref{thm:cineq} provides a concentration bound on the solution to $L(\mech)$ given the solution to $\hat{L}(\mech)$. \\
\textit{Remarks}:
\begin{enumerate}[label=\roman*)]
\item In practice we will only have an empirical approximation $\hat{\msb}$ of the full densities $\msb$, yet the mechanism design goal is still to construct a mechanism making the \textit{full} (unobserved) densities socially optimal. Theorem~\ref{thm:cineq} is thus useful in that it allows us to quantify the degree of suboptimality of the \textit{full} strategies, given through \eqref{eq:ineq_aug}, given a certain degree of suboptimality of these \textit{partially observed} strategies, given through \eqref{eq:ineq_aug2}. We will use this in the following section when formulating our algorithmic convergence guarantees.

\item In Section~\ref{sec:robustrp}, we provide an equivalence between the Pareto gap loss function \eqref{eq:ineq_aug} and several well-established relaxed revealed preference metrics. This will allow us another principled way to quantify this "degree of suboptimality" mentioned above.
\end{enumerate}

\textcolor{black}{In this section we built off of the linear programming test for social optimality consistency, of Theorem~\ref{thm:MA_Af}, to construct a "Pareto gap" loss function which quantifies the "distance" from social optimality a dataset of observed equilibria strategies are. We also constructed a stochastic Pareto gap which augments the Pareto gap to handle finite-sample strategies. The former represents the idealistic full-information case, which we use as a model for proving consistency results, and the latter is the practically relevant construction which will be used in Algorithm~\ref{alg:amd}.}

\subsection{Simultaneous Perturbation Stochastic Approximation for RL}
\label{sec:spsa}

Here we describe a simultaneous perturbation stochastic approximation (SPSA) algorithm which operates on the stochastic loss function \eqref{eq:ineq_aug2}. We operate on \eqref{eq:ineq_aug2} rather than the deterministic loss function \eqref{eq:ineq_aug} for the sake of generality, and since in practice we will only have access to a finite number of i.i.d. samples from mixed-strategies.\textcolor{black}{The proposed algorithm iteratively adjusts the mechanism parameter $\mech$ by observing batch datasets $\CD_{\mech}$, evaluating $\hat{L}(\mech)$, and moving $\mech$ in the direction of decreasing $\hat{L}(\mech)$ by forming a stochastic gradient approximation, until a minima $\hat{\mech} \in \arg\min_{\mech \in \mechspace}\hat{L}(\mech)$ is obtained}.  

We adopt a decreasing step-size SPSA algorithm, which is able to globally optimize non-convex functions \cite{maryak2001global}. With objective function $L: \reals^p \to \reals$, this algorithm iterates via the recursion
\begin{equation}
\label{eq:spsa}
    \mech_{k+1} = \mech_k - a_k G_k(\mech_k) + q_k w_k,\,\,\, k=1,2,\dots
\end{equation}
where $w_k \overset{i.i.d.}{\sim} \gaussN(0,I)$ is purposefully injected Gaussian noise, $a_k = a/k, q_k^2 = q/k\log\log(k), a>0, q>0$, and $G_k$ is a term approximating the function gradient, 
\[G_k = \frac{\hat{L}(\mech_k + c_k\Delta_k) - \hat{L}(\mech_k - c_k\Delta_k)}{2c_k\Delta_k}\]
where $\hat{L}$ is the stochastic Pareto gap loss function \eqref{eq:ineq_aug2}.
Here $c_k$ is a scalar step-size for determining the finite-difference gradient, and $\Delta_k = [\Delta_{k_1},\dots,\Delta_{k_p}]'\in \reals^p$ is a $p$-dimensional $\pm 1$ random variable, i.e., for each $i\in[p]$, $\Delta_{k_i}$ is an independent Bernoulli(1/2) random variable. 

\begin{algorithm}[h!]
\caption{Adaptive Mechanism Design}
\label{alg:amd}
\begin{algorithmic}[1]
\State \textcolor{black}{Initialize mechanism parameter $\mech_0 \in \mechspace \subset \reals^p$, dataset size $T\in\nat$, algorithmic parameters $a,c,q>0, \eta \in[\frac{1}{6},\frac{1}{2}]$, and mechanism-dependent action space-inducing constraints $\{g_{\mech}^i(\cdot)\}_{i=1}^{\na}$}
\State $\itr=0$;

\For {$t=1:T$}
    \For {$i=1:\na$}
        \State generate $a_t, \beta^i$ and $g_{\mech,t}^i(\cdot)$ as in Assumption~\ref{as:confun} on page~\pageref{as:confun};
    \EndFor
\EndFor
\While{optimizing}
\State $\itr = \itr+1$;
\State $\Delta_{\itr} \sim \textrm{$p$-dimensional  $\pm 1 $ vector}$
\State $c_{\itr} = c/\itr^{\eta}; \, a_{\itr} = a/\itr; \, q_{\itr} = \sqrt{q/(\itr\log(\log(\itr)))};$
\State $\mech_{\itr,1} = \mech_{\itr} + c_{\itr} \Delta_{\itr}$;$\,\,\,\mech_{\itr,2} = \mech_{\itr} - c_{\itr} \Delta_{\itr}$;
\For{$\zeta=1:2$}
\For{$t=1:T$}
    \State enact game $G$ with $\mech_{\itr,\zeta}$,  $\{g_{\mech,t}^i(\cdot)\}_{i\in[\na]}$;
    \State observe empirical mixed-strategy Nash equilibrium distribution $\hat{\msb}_t(\mech_{\itr,\zeta})$;
\EndFor
\State solve: minimize $r$ s.t.$\, \exists \{u_t^i \in \reals,\lambda_t^i > \alpha\}$:\\
$\quad \quad \quad \quad \quad \quad u_s^i - u_t^i - \lambda_t^i\nlconst^i(\hat{\msb}_s(\mech_{m,\zeta})) \leq \lambda_t^i r \, \, \, \,\forall t,s,i$;  \eqref{eq:ineq_aug2}
\State set $r_{\zeta} = r$;
\EndFor
    \State set $g_{\itr} = [(g_{\itr})_1,\dots,(g_{\itr})_p]$ with 
    \State \quad \quad $(g_{\itr})_i = (r_1 - r_2)/(2c_{\itr}(\Delta_{\itr})_i)$;
        \State $\mech_{\itr+1} = \textrm{Proj}_{\mechspace}(\mech_{\itr} - a_{\itr} g_{\itr} + q_{\itr} w_{\itr})$; where $ w_{\itr} \sim \gaussN(0,I)$ and $\textrm{Proj}$ is defined in \eqref{eq:projdef}
\EndWhile
 
\end{algorithmic}
\end{algorithm}
\textit{Remarks}:
\begin{enumerate}[label=\roman*)]
\item Observe that this algorithm only evaluates the objective function itself and does not rely on access to its gradients. This is appealing for use with our objective function \eqref{eq:ineq_aug} since it is not immediately obvious how one could access or compute the gradient of \eqref{eq:ineq_aug}. 

\item Algorithm \eqref{eq:spsa} purposefully injects Gaussian noise and thus imitates a simulated annealing \cite{KY03} or stochastic gradient Langevin dynamics \cite{welling2011bayesian} procedure. It turns out that this allows the algorithm to converge to \textit{global} minima of potentially nonconvex functions by escaping local minima \cite{maryak2001global}.
Indeed, we will provide a result guaranteeing the convergence of our algorithm to the global minima of $L(\mech)$. 
\end{enumerate}

\paragraph{Implementation} The SPSA  algorithm within our adaptive mechanism design framework is presented  in Algorithm~\ref{alg:amd}. In this implementation, an estimate $\hat{L}(\mech)$ of the Pareto gap $L(\mech)$ is evaluated by obtaining the empirical dataset $\CD_{\mech}$, with finite-sample empirical distribution \eqref{eq:empdist}, and solving the linear program \eqref{eq:ineq_aug2}. Recall the empirical dataset $\CD_{\mech}$ is obtained by "probing" the multi-agent system with constraint functions $\{g_{\mech,t}^i(\cdot)\}_{i=1}^{\na}$ over $T$ time points, and observing the mixed-strategy Nash equilibrium responses $\{\hat{\msb}_t(\mech),t\in[T]\}$ for mechanism $o_{\mech}$. We assume that the mechanism parameters are taken from a compact space $\mechspace \subset \reals^p$, and so whenever an iterate $\mech_{k+1}$ leaves this space, we project it back onto the boundary of $\mechspace$ by the function 
\begin{equation}
\label{eq:projdef}
    \textrm{Proj}_{\mechspace}(\mech) = \arg\min_{\gamma\in\mechspace}\|\gamma-\mech\|
\end{equation}

\subsection{SPSA Convergence}
\label{sec:convg}

\textcolor{black}{Algorithm~\ref{alg:amd} optimizes the stochastic Pareto gap loss function $\hat{L}(\mech)$ using SPSA; in Section~\ref{sec:convg} we will prove that this algorithm is asymptotically consistent, in that it globally minimizes the Pareto gap $L(\mech)$ as the number of empirical samples $N$ and the algorithmic run-time $m$ tend to infinity.}
 We require  the following  assumptions.

\subsubsection{Assumptions}
\label{sec:as}

\begin{assumption}[Utility Function Structure]
\label{as:pay}
The utility functions $\{f^i_{\mech}(x)\}_{i=1}^{\na}$, with $x= [x_1,\dots,x_{\na}] \in \reals^{k\na}$, are each concave and increasing with respect to $x_i$, thrice continuously differentiable with bounded $p$-th order derivatives ($p=1,2,3$) with respect to both $x_i$ and $\mech$. 
\end{assumption}

\begin{assumption}[Constraint Function Regularity]
\label{as:con}
The constraint functions $g_{\mech,t}^i: \CA^i \to \reals$ are increasing and thrice continuously differentiable for all $t$ and $i$. 
\end{assumption}

\begin{assumption}[Mechanism Parameter Space]
\label{as:pspace}
The mechanism parameter space $\mechspace$ is a non-empty  subset of $\reals^p$ with  positive Lebesgue-measure, and there exists a subset $\mechspace_c \subseteq \mechspace$ such that each $f_{\mech}^i(\cdot)$ is continuous on $\reals^{k\na}$ for every $\mech \in \mechspace_c$. 
\end{assumption}

\textit{Discussion of Assumptions}: By \cite{rosen1965existence}, concavity of utility functions $\{f^i_{\mech}(\cdot)\}_{i=1}^{\na}$ and convexity of the constraint regions defined by $\{g_{\mech,t}^i(\cdot)\}_{i=1}^{\na}$ is sufficient to ensure existence of a pure-strategy Nash equilibrium solution, and thus of a mixed-strategy Nash equilibrium for every $\mech\in\mechspace$. Concavity (or quasi-concavity) of utility functions is widely assumed in the microeconomic and game theory literature \cite{fudenberg1991game}. Each $\CA_{\mech,t}^i$ is compact since it exists in the non-negative orthant of $\reals^k$, and is bounded above by $\{x\in\reals^k : g_{\mech,t}^i(x)=0\}$. Furthermore, the additional differentiability structure imposed by Assumptions \ref{as:pay} and \ref{as:con} will be sufficient to guarantee convergence of Algorithm~\ref{alg:amd} to the set $\hat{\mechspace} = \arg\min_{\mech\in\mechspace}L(\mech)$ of global minimizers of $L(\mech)$. \textcolor{black}{However, we emphasize that this differentiability may not be a \textit{necessary} condition for this convergence, and thus our approach may be empirically effective even for non-smooth utilities.\footnote{\textcolor{black}{A comprehensive account of \textit{necessary} conditions for Algorithm~\ref{alg:amd} convergence is beyond the scope of this work, but would make for interesting future study.}}}

\subsubsection{Main Result II. SPSA Convergence} Our second main result is that Algorithm~\ref{alg:amd} converges in probability to the set of global minimizers of $L(\mech)$. 
\begin{theorem}
\label{thm:convg}
   \textcolor{black}{Let $\hat{\mechspace} = \arg\min_{\mech\in\mechspace}L(\mech)$ be the set of global minimizers of $L(\mech)$} and $\|\cdot\|$ denote $L^2$ norm. Recall $k$ is the iteration number of Algorithm~\ref{alg:amd}, and $N$ is the number samples making up the empirical distribution \eqref{eq:empdist}. Under Assumptions \ref{as:pay}-\ref{as:pspace}, we have 
    \begin{align*}
        &\forall \epsilon>0,\, \lim_{k,N\to\infty}\PR(d(\mech_k, \hat{\mechspace}) > \epsilon) = 0,  \quad  \textrm{ where } \,\, d(\mech_k,\hat{\mechspace}) = \min_{\hat{\mech}\in\hat{\mechspace}}\|\mech_k-\hat{\mech}\|
    \end{align*}
    i.e., Algorithm 1 converges in probability to the set of global minimizers of $L(\mech)$
\end{theorem}

\begin{proof}
    See Appendix~\ref{sec:con_pf}
\end{proof}

{\color{black}
This global minimization produces a mechanism $\hat{\mech}$ which is the \textit{closest} feasible mechanism to social optimality, and has the following key interpretation as an \textit{identification} method for whether Pareto-optimality is fundamentally achievable for the game setup at hand.

\begin{corollary}
Theorem~\ref{thm:convg} guarantees that Algorithm~\ref{alg:amd} converges in probability to a \textit{global} minimizer of $L(\mech)$. This ensures that:
 \begin{enumerate}[label=\arabic*)]
    \item If $\min_{\mech\in\mechspace}L(\mech) = 0$,
    Algorithm~\ref{alg:amd} obtains a socially optimal mechanism $\mech^*$ satisfying  $\mu(\mech^*) \in \nashe(\mech^*) \Rightarrow \mu(\mech^*)\in \soco(\mech^*)$.
    \item If $\min_{\mech\in\mechspace}L(\mech) > 0$, Algorithm~\ref{alg:amd} produces the mechanism \textit{closest} to social optimality, and it is identified that no mechanism inducing social optimality exists.
    \end{enumerate}
Thus, a feature of our approach is that we can \emph{identify} whether Pareto-optimality is a \emph{possible} outcome for a designer, and if so Algorithm~\ref{alg:amd} automatically finds the mechanism inducing it.
\end{corollary}}

\paragraph{Comparison with alternative optimization methods.}
\textcolor{black}{While the SPSA recursion provides the theoretical convergence guarantees necessary for our setting, it is worth noting that other optimizers, such as Evolutionary Strategies (ES), Bayesian Optimization (BO), or stochastic gradient Langevin dynamics (SGLD), could in principle be employed to minimize the Pareto gap. Each offers different trade-offs in terms of sample efficiency and computational cost. ES methods \cite{hansen2015evolution} approximate the gradient via finite-population perturbations and are often effective in low-noise continuous control problems, but their population-based sampling requires 
$O(p)$, where $\mech\in\mechspace\subset\reals^p$, or more function evaluations per iteration. In our mechanism-design setting, each function evaluation entails an equilibria strategy observation, making such methods substantially more expensive than SPSA’s constant-two-query structure. Bayesian optimization approaches \cite{shahriari2016unbounded} can be highly sample-efficient in low-dimensional problems but scale poorly with 
$p$ due to Gaussian-process kernel inversion and model-update costs. SGLD global minimization guarantees have been analyzed in \cite{raginsky2017non}.}

\subsection{Robust RL through Relaxed Revealed Preferences}
\label{sec:robustrp}
\textcolor{black}{Theorem~\ref{thm:convg} states that Algorithm~\ref{alg:amd} globally minimizes the Pareto gap $L(\mech)$; achieving social optimality if the level set $L(\mech)=0$ is non-empty. If this level set is empty, Algorithm~\ref{alg:amd} will \textit{identify} the inherent unachievability of social optimality for the game setup at hand, by converging to a mechanism $\mech$ inducing $L(\mech) > 0$. We claim that this non-zero Pareto gap convergence attains a mechanism which is the "closest" achievable to social optimality. The purpose of this section is to provide a rigorous sense for this social optimality proximity measure, giving meaning to our usage of "closest". }


Specifically, we provide an equivalence between the Pareto gap value $L(\mech)$ and several well-established revealed preference goodness-of-fit metrics, namely "relaxed revealed preferences". Thus, if a mechanism $\mech$ induces a positive Pareto gap value $L(\mech)$ (such that the dataset $\dataset_{\mech}$ is \textit{not} consistent with social optimality), we may reason about its performance in terms of these metrics.


Robust revealed preference metrics are well-established in the microeconomic literature. Afriat \cite{afriat1972efficiency} and Varian \cite{varian1990goodness} argued that full consistency with optimality may be too restrictive hypothesis in empirical applications, and thus introduced several metrics to capture "nearly optimal" behavior. These notions are typically quantified through relaxations of GARP. Thus, our equivalence result allows for an understanding of our algorithmic suboptimality through the lens of these developed frameworks. 


\subsubsection{Relaxed Revealed Preference Metrics}
In order to capture these metrics in the regime of our nonlinear budget constraint, we may here assume that we work with budget constraints $\bg_{\mech,t}^i$ which are satiated at one, $\bg_{\mech,t}^i(\msi_t)\leq 1$, rather than zero, without loss of generality. These can simply be formed as $\bg_{\mech,t}^i(\cdot) = g_{\mech,t}^i(\cdot)+1$, and are necessary so that multiplicative factors at satiation are non-degenerate.  

Perhaps the most standard of the relaxed metrics is the Critical Cost Efficiency Index (CCEI), introduced in \cite{afriat1972efficiency}. The CCEI augments GARP by a vector $e^i = \{e_1^i,\dots, e_T^i\}$ with elements $e_t^i \in [0,1]$, such that $\msi_t R_e \msi_s$ if $e_t^i \bg_{\mech,t}^i(\msi_t) \geq \bg_{\mech,t}^i(\msi_s)$. We say that GARP$_{e^i}$ holds if $\msi_t H_e \msi_s$ implies $e_s^i \bg_s^i(\msi_s) \leq \bg_s^i(\msi_t)$, where $H_e$ is the transitive closure of the relation $R_e$. The CCEI for agent $i$ is then defined as the largest possible average of $e^i$, i.e.,
\begin{align*}
    CCEI(\CD_{\mech},i) = \max_{ e^i \in[0,1]^T}\frac{1}{T}\sum_{t\leq T}e_t^i \quad s.t. \,\, \CD_{\mech} \,\, \textrm{satisfies} \,\, \textrm{GARP}_{e^i} 
\end{align*}
Intuitively, CCEI$(\CD_{\mech},i)$ represents the smallest proportion of consumer $i's$ budget which she is allowed to waste through suboptimal consumption. Also, CCEI reduces to GARP \cite{forges2009afriat} when $e^i=\boldsymbol{1}$. CCEI is commonly used in empirical microeconomic applications, and there exist polynomial time algorithms for implementing CCEI \cite{hjertstrand2013simple}. 

Now, we relate the CCEI to another relaxed metric, GARP$_F$ introduced in 
\cite{hjertstrand2013simple}.

\begin{definition}[GARP$_F$ \cite{hjertstrand2013simple}]
\label{def:garpf}Consider a dataset $\CD_{\mech}$, and denote
\begin{enumerate}[label=\roman*)]
    \item $\msi_k\, R_F \, \msi_j$ if $\,\, \bg_k^i(\msi_j) + F\leq \bg_k^i(\msi_k)$
    \item $H_F$ the transitive closure of $R_F$
\end{enumerate}
We say $D_{\mech}$ satisfies GARP$_F$ if $\msi_k \, H_F\, \msi_j$ implies $\bg_j^i(\msi_j) \leq \bg_j^i(\msi_k) + F$ for all $i\in[\na]$.
\end{definition}
Observe that GARP$_F$ reduces to GARP \cite{afriat1967construction} when $F=0$. Like the CCEI, GARP$_F$ measures the smallest perturbation of the dataset $\CD_{\mech}$ such that the data is consistent with optimality. While this perturbation is multiplicative in the CCEI, it is additive in GARP$_F$. Also notice that the CCEI can be recovered from GARP$_F$ by letting $e_t^i = (1- \bg_{\mech,t}^i/F)$. We introduce GARP$_F$ because we can show straightforwardly that it is \textit{equivalent} to our Pareto gap metric $L(\mech)$ \eqref{eq:ineq_aug}. Specifically, we have the following relation

\begin{theorem} 
\label{thm:rrp} Recall the Pareto gap $L(\mech)$ \eqref{eq:ineq_aug} and $GARP_F$ in Definition~\ref{def:garpf}. The following are equivalent
\label{thm:robusteq}
\begin{enumerate}[label=\roman*)]
    \item $L(\mech) = F$
    \item GARP$_F$ holds
\end{enumerate}
\end{theorem}

\begin{proof}
    This follows from direct extension of Theorem 2 in \cite{hjertstrand2013simple}. 
\end{proof}

\textit{Discussion}: Theorem~\ref{thm:robusteq} gives us a way of relating the value of the Pareto gap $L(\mech)$ to the relaxed revealed preference metrics. Thus, by operating on $L(\mech)$ in Algorithm~\ref{alg:amd}, we also control the proximity of dataset $\CD_{\mech}$ to social optimality in the sense of these well-established principles.\textcolor{black}{The key idea is that we know  $L(\mech) = 0$ indicates consistency with social optimality, but we now can quantify the "proximity" to social optimality achieved by $L(\mech) > 0$ through the equivalence to GARP$_F$ and the CCEI. In particular, we now have a framework for reasoning about the ordinal relations between positive values of $L(\mech)$: if $L(\mech_1) > L(\mech_2) > 0$, we can conclude that $\mech_2$ is "closer" to inducing social optimality than $\mech_1$ in the sense of CCEI and GARP$_F$. }

\subsection{Individual Rationality Guarantees and Repeated-Game Generalizations}
{\color{black} 
Here we present discussions for extending the present framework in two pertinent directions: to guarantee individual rationality conditions, and to generalize beyond the current iterated stage-game structure to repeated games by removing stage-stationarity of the Pareto gap.
\subsubsection{Participation and Individual Rationality}
In classical mechanism design, \emph{individual rationality} (IR) ensures that under equilibrium, each agent weakly prefers
participation over abstention, i.e. 
\begin{equation}
\label{eq:IR}
    \CE_{(a_i,a_{-i}) \sim \mu(\phi)}[f_{\phi}^i(a_i, a_{-i})] \ge c_i,
    \quad \forall i,
\end{equation}
where $c_i$ is agent $i$'s utility of non-participation. This notion of individual rationality is readily incorporated into our framework, and can be imposed through the following structural condition:

\begin{assumption}[Non-Participation Invariance]
\label{as:IR}
For every agent $i$, there exists a feasible action $\bar{a}_i\in\CA_{\phi,t}^i$ for all $t\in[T]$, such that 
\begin{equation}
\label{eq:ir_as}
f_{\phi}^i(\bar{a}_i,a_{-i}) = c_i\,\, \forall \phi\in\Phi, a_{-i} \in \CA_{\phi,t}^{-i}
\end{equation}
where $c_i \in \reals$ is the utility of non-participation. 
\end{assumption}

Thus, we require that there exists \emph{at least one} feasible action which induces the non-participation utility, for any mechanism $\phi$ and any opponent joint-action $a_{-i}$. Consequently,  the utility $f_{\phi}^i(\cdot,a_{-i})$ under $\bar{a}_i$ is \emph{invariant} to the mechanism $\phi$ and opponent joint-action $a_{-i}$. This assumption is  naturally motivated in  engineering and economics  applications, imposes the individual rationality condition, and does not alter the technical results of this paper:
\begin{enumerate}
    \item \emph{Applicability.}  Assumption~\ref{as:IR} holds when $f_\phi^i(a_i,a_{-i}) = u_i(a_i, a_{-i};\phi) - k_i(a_i,a_{-i};\phi)$ with an interaction benefit $u_i$ and cost $k_i$ that vanish at some feasible $\bar{a}_i$ independently of $a_{-i}$ or $\phi$. This is standard for mechanisms governing resource allocation or coordination, where "non-participation" has an unambiguous physical or economic meaning. For instance,  the wireless spectrum sharing example in Section~\ref{sec:netex} and the river pollution game in Section~\ref{sec:numsim}, fall into this category and thus satisfy Assumption~\ref{as:IR}, by choosing non-participation as $\bar{a}_i=0$. Then $c_i = 0$, which is  naturally interpreted as the utility of non-participation.  
    \item \emph{Individual rationality implication.} \eqref{eq:IR} holds at equilibrium, since if agent $i$'s equilibrium strategy has expected payoff below $c_i$, deviating to $\bar{a}_i$ would strictly increase expected utility, contradicting equilibrium. Hence, $\CE_{a_i \sim \mu^i(\phi), a_{-i}\sim \mu^{-i}(\phi)}[f_{\phi}^i(a_i,a_{-i})] \geq c_i$ in the equilibria mixed strategy $\mu^i(\phi)$, recovering \eqref{eq:IR} in expectation for every $\phi\in \Phi$.

    \item \emph{Compatibility with technical results.} Allowing $\bar{a}_i \in \CA_{\phi,t}^i$, leaves the action space regularity of Assumptions~\ref{as:confun}, \ref{as:con} unchanged, and thus Theorem~\ref{thm:MA_Af} is unaffected.  Compatibility with Theorem~\ref{thm:convg} follows from the smoothness properties  required by Assumption~\ref{as:pay}; the differentiability arguments underlying Theorem~\ref{thm:convg} then apply unmodified.
\end{enumerate}

\subsubsection{Beyond stage-stationarity of the Pareto gap}
\label{rem:stage-stationarity}
The convergence guarantee of Algorithm~1 (Theorem~2) relies on the fact that,
for a fixed mechanism $\phi$, the induced Pareto gap $L(\phi)$ is independent of the stage index $m$ and independent of the history of past play. This holds because the stage-wise equilibrium map $\phi_m \mapsto \mu_t(\phi_m)$ is \emph{stage-stationary}: the distribution of $\hat{\msb}_t(\mech_{\itr,\zeta})$ conditional on $\phi_{m,\zeta} = \phi$ does not depend on $m$ or on the filtration $\mathcal{F}_{m-1}$ generated by prior mechanisms and observed equilibria. This is due to the isolation of incentives (constrained utilities) to the structure of each distinct stage-game \eqref{eq:Gmixed} in our model, and is precisely what allows $L(\phi)$ to be treated as a fixed, time-invariant objective landscape that Algorithm~1's stochastic approximation recursion can be shown to converge within.

If agents were permitted to strategize with respect to the trajectory of mechanisms $(\phi_1,\dots,\phi_m)$ rather than $\phi_m$ alone (e.g., to build reputation or anticipate future updates \cite{fudenberg1990evolution}), the equilibrium map would
become history-dependent and anticipative. The Pareto gap then inherits this dependence and becomes a stochastic, $\mathcal{F}_{m-1}$-measurable object $L(\phi, \mathcal{F}_{m-1})$ rather than a fixed function of $\phi$. Theorem~\ref{thm:convg} does not cover optimization of such a history-dependent non-stationary objective; establishing convergence in that setting would require substantially different tools; we view this as a distinct, promising direction for future work.}

\section{Finite Sample Consistency via Distributionally Robust RL}
\label{sec:psdro}
Thus far, our approach achieves \textcolor{black}{the mechanism design objective} in the asymptotic limit as the number of i.i.d. samples tends to infinity (Theorem~\ref{thm:convg}). This section considers the case when only a \textit{finite-sample} empirical distribution is available. We employ distributionally robust RL as a principled approach to handle this case, in order to provide guarantees for the full strategies from only finite empirical distributions.  
Distributionally robust optimization is an approach that seeks optimal decision-making under uncertainty by minimizing the worst-case scenario across a set of possible probability distributions, ensuring solutions are resilient to variations in data. We utilize this to minimize the Pareto gap over the worst-case distribution in proximity of the empirical distribution, thereby minimizing the Pareto gap of the true fully specified distributions. Our contributions are:
\begin{enumerate}[label=\roman*)]
    \item We formulate a distributionally robust objective function which allows us to achieve mechanism design for the \textit{full} (unobserved) mixed strategies from a finite set of $i.i.d.$ samples.
    \item We prove the equivalence between this objective function and a reformulated semi-infinite program, allowing us to employ exchange methods to compute solutions. 
    \item We develop an algorithmic solution, provide convergence rate guarantees, and discuss its interface with Algorithm~\ref{alg:amd}. In Section~\ref{sec:numsim} we provide a numerical example demonstrating this algorithm's rapid convergence.
\end{enumerate}

\subsection{Distributionally Robust Objective as a Means for Mechanism Design}
We formulate a distributionally robust optimization (DRO) problem which will be used to achieve finite-sample consistency. \textcolor{black}{In practice, observed play from equilibria strategies can only form a finite-sample approximation of the true underlying mixed-strategies. Thus, while the asymptotic consistency result of Theorem~\ref{thm:convg} is a useful theoretical benchmark, optimizing the Pareto gap in practice will necessitate quantification of finite-sample performance. Here we formulate a distributionally robust optimization problem as a means to minimize the Pareto gap $L(\mech)$ from only finite-sample empirical distributions. }

\subsubsection{Notation}
 In the partially observed case (recall notation in Algorithm~\ref{alg:amd}) we have  $\hat{\msb}_t(\mech_{n,a},\cdot)$ forms an empirical distribution on $\Delta(\CA_{\mech,t})$. We can extend this domain to the tensor product $\otimes_{t=1}^{T} \Delta(\CA_{\mech,t})$, such that an element in this space is the product $\prod_{t=1}^T \msb_t$ of mixed-strategies (distributions) on $\Delta(\CA_{\mech,t})$. Here we denote the product measure over $[T]$ as $\prod_{t=1}^T \mu_t$, to distinguish it from the agent-product measure $\mu_t := \otimes_i \mu_t^i$ introduced in previous sections — the former ranges over time indices $t\in[T]$, the latter over agents $i\in[n]$ at fixed $t$. We identify this product $\prod_{t=1}^T \msb_t$ with the sequence $\{\msb_t(\mech_{n,a},\cdot),t\in[T]\}$ of mixed strategies. We denote a sample from probability distribution $\prod_{t=1}^T \msb_t(\mech_{n,a},\cdot)$ as $$\Theta = \prod_{t=1}^T \gamma_t \in \Gamma := \otimes_{t=1}^T \CA_{\mech,t}. $$ We also introduce the following notation: $\psi = [u_1^1,\lambda_1^1,\dots,u_T^{\na},\lambda_T^{\na}] \in \Psi$, where
{\color{black}\begin{equation}
\label{eq:psidef}
\Psi = \{[u_1^1,\lambda_1^1,\dots,u_T^{\na},\lambda_T^{\na}]: u_t^i \in \reals, \lambda_t^i > 0 \,\, \forall t\in[T], i\in[\na] \}
\end{equation}}

\subsubsection{Distributionally Robust Optimization}
Now we construct a loss function $\hat{L}^{\epsilon}_{DR}(\mech)$ as the following \textit{distributionally robust optimization} problem:
\begin{align}
\begin{split}
\label{eq:dro}
    &\hat{L}_{DR}^{\epsilon}(\mech) = \min_{\psi \in \Psi} \sup_{Q\sim B_{\epsilon}(\prod_{t=1}^T\hat{\msb}_t)}\CE_{\Theta \sim Q}\left[h(\psi,\Theta)\right] \\
    & h(\psi,\Theta) := \arg\min_{r\geq 0} : u_s^i - u_t^i - \lambda_t^i g_{\mech,t}^i(\gamma_s^i) \leq \lambda_t^i\,r \quad\forall s,t\in[T],i\in[\na]
\end{split}
\end{align}
where $B_{\epsilon}(\prod_{t=1}^T\hat{\msb}_t)$ is the set of distributions on $\otimes_{t=1}^T\Delta(\CA_{\mech,t})$ with 1-Wasserstein distance at most $\epsilon$ from the empirical distribution $\prod_{t=1}^T\hat{\msb}_t$. The 1-Wasserstein distance between distributions $Q$ and $P$ on space $\CX$ is given by
    \[\wass(Q,P) = \inf_{\pi\in\Pi(Q,P)}\int_{\CX \times \CX} \| x - y\|_2 \, \pi(dx,dy),\]
    where $\Pi(Q,P)$ is the set of probability distributions on $\CX \times \CX$ with marginals $Q$ and $P$. The Wasserstein metric is useful in dealing with empirical measures because it does not require absolute continuity (unlike e.g., KL divergence).

Then we have the following result relating the distributionally robust objective \eqref{eq:dro}, \textit{formed with only $N$ empirical mixed-strategy samples}, to the loss function \eqref{eq:ineq_aug} operating on \textit{fully-specified mixed strategies}. 

\begin{theorem}
\label{thm:dreq}
{\color{black} Recall $L(\mech)$ \eqref{eq:ineq_aug}. Let $W_{\epsilon}$ be the distributional ball centered at true mixed strategies $\prod_{t=1}^T\msb_t$, with $\epsilon$ 1-Wasserstein radius. Suppose the empirical distribution $\prod_{t=1}^T\hat{\msb}_t \in W_{\epsilon}$, that is,  $\wass(\prod_{t=1}^T\hat{\msb}_t, \prod_{t=1}^T\msb_t) \leq \epsilon$. Then for any $c\in\reals$,}
\begin{equation}
\label{eq:drimp}
\hat{L}^{\epsilon}_{DR}(\mech) \leq  c \Rightarrow L(\mech) \leq c 
\end{equation}
\end{theorem}
\begin{proof}
See Appendix~\ref{sec:pfdreq}.
\end{proof}

\textcolor{black}{Thus, Theorem~\ref{thm:dreq} implies that we can \textit{control $L(\mech)$}, as specified by the true unobserved strategies, by operating on $\hat{L}^{\epsilon}_{DR}(\mech)$, formed only by a finite-sample approximation of the true strategies.}
\textit{Remark}: In reality we do not know the true mixed strategies $\prod_{t=1}^T\msb_t$, and thus cannot guarantee the $\epsilon$ proximity necessary for \eqref{eq:drimp}. We provide \eqref{eq:drimp} as a consistency result. Theorem~\ref{thm:dreq} serves as an important step in generating probabilistic guarantees from the underlying law governing the Wasserstein proximity between these distributions. Such techniques can be found in e.g., \cite{dedecker2017behavior}. For brevity we leave this to future work. 

The importance of this procedure is that it provides the following guarantee:
\begin{lemma}
\label{lem:drogr}
    Suppose the implication \eqref{eq:drimp} holds, and $\wass(\prod_{t=1}^T\hat{\msb}_t, \prod_{t=1}^T\msb_t) \leq \epsilon$. \textcolor{black}{Recall $\hat{\mechspace} = \arg\min_{\mech\in\mechspace}L(\mech)$. Then, \textit{using $\hat{L}^{\epsilon}_{DR}(\mech)$ in lines 18-20 of Algorithm~\ref{alg:amd}, we converge in probability to $\hat{\mechspace}$ for $L(\mech)$ formed from the full mixed strategies, from only i.i.d. samples (partial specifications)}}. Specifically,
    \[\forall \epsilon>0,\, \lim_{k\to\infty}\PR(d(\mech_k, \hat{\mechspace}) > \epsilon) = 0,  \quad  \textrm{ where } \,\, d(\mech_k,\hat{\mechspace}) = \min_{\hat{\mech}\in\hat{\mechspace}}\|\mech_k-\hat{\mech}\|\]
    where we \textit{no longer need the number of samples $N$ to tend to infinity}.
\end{lemma}

As a consequence of Lemma~\ref{lem:drogr}, even though $\arg\min_{\mech\in\mechspace} L(\mech)$ is determined by the full mixed-strategies $\prod_{t=1}^T \msb_t$, we can exploit the distributionally robust optimization \eqref{eq:dro} to achieve $\mech \in \arg\min_{\mech\in\mechspace} L(\mech)$ from only finite empirical approximations $\prod_{t=1}^T \hat{\msb}_t$.



\subsection{Distributionally Robust Solutions via Semi-Infinite Optimization}

  Here we provide an equivalence between the distributionally robust objective \eqref{eq:dro} and a particular semi-infinite program. \textcolor{black}{This formalism reformulates the objective \eqref{eq:dro} into a structure which allows for computational approximation of solutions. In particular, we may compute approximations of solutions to \eqref{eq:dro} \footnote{A $\delta$-optimal solution of \eqref{eq:siprog} is a solution of the same problem with constraint bounds relaxed by $\delta>0$} by exploiting exchange methods to solve the semi-infinite program.} 
  \subsubsection{Preliminary Specifications and Notation} \textcolor{black}{We introduce the following specification and notation. These constructions will be crucial for ensuring the validity of our following semi-infinite program reformulation (Theorem~\ref{thm:sireform}). See \cite{luo2017decomposition} for more details on these specifications and their necessity in forming an equivalence between distributionally robust optimization and semi-infinite programming.}
    \begin{specification}[Parameter Set Bounds]
    \label{as:convex}
        {\color{black} Recall $\Psi$ as defined in \eqref{eq:psidef}. Here we restrict $\Psi$ to the compact set specified by $\{[u_1^1,\lambda_1^1,\dots,u_T^{\na},\lambda_{T}^{\na}]\}$ with $u_s^i \in [-1,1],\, \lambda_s^i \in [\hat{\lambda},1], \,\forall s \in[T],i\in[\na]$, for some $\hat{\lambda}>0$. This is without loss of generality.}\footnote{Specification~\ref{as:convex} is not listed as an assumption because it is without loss of generality. Observe: if a set of parameters $\psi = [u_1^1,\dots,\lambda_T^{\na}]\in \Psi$ solves $u_s^i - u_t^i - \lambda_t^i g_{\mech,t}^i(\gamma_s^i) \leq \lambda_t^i\,r$, then so does $c\,\psi := [cu_1^1,\dots, c\lambda_T^{\na}]$ for any scalar $c>0$. Also, given the boundedness of $\|\alpha_t\|$ and $\|\beta_t^i\|$ the ratio $u_s^i/\lambda_t^i$ will be bounded from above and below by positive real numbers. Thus, we can always find some $\psi$ solving $u_s^i - u_t^i - \lambda_t^i g_{\mech,t}^i(\gamma_s^i) \leq \lambda_t^i\,r$ such that $u_s^i \in [-1,1], \, \lambda_s^i \in [\hat{\lambda},1], \forall s\in[T], i\in[\na]$, for some $\hat{\lambda}>0$.}
    \end{specification}

     By Specification~\ref{as:convex} we must have that $h(\psi,\Theta) \leq V := \frac{2}{\hat{\lambda}}$ for any $\psi \in \Psi, \, \Theta \in \Gamma$. Let us denote \[\CV := \biggl\{\bv \in \reals^{N+1}: \,\,0 \leq v_i \leq 2V \,\, \forall i \in [N],\, \,0\leq v_{N+1}\leq V/\epsilon\biggr\}\] Also observe that by Assumption~\ref{as:confun}, we need only consider constraint sets $\CA_{\mech,t}^i = \{x\in\reals_+^k : G \leq g_{\mech,t}^i(x) \leq 0\}$, with $G := |\inf_{i,t} g_{\mech,t}^i(\boldsymbol{0})|$ as in Theorem~\ref{thm:cineq}. Thus, $\Gamma$ is effectively a compact set. 
     
     \subsubsection{Semi-Infinite Program Reformulation} 
     The following result provides an equivalence between the distributionally robust optimization problem \eqref{eq:dro} (used for achieving mechanism design for \textit{full} strategies from only \textit{partial} strategy specifications) and a semi-infinite program. \textcolor{black}{A semi-infinite program is an optimization problem with a finite number of variables to be optimized but an arbitrary number (continuum) of constraints; we then apply a finite exchange method to obtain approximate solutions of this program}.

    \begin{algorithm}[b]
    \caption{Wasserstein Robust Utility Estimation}
    \label{alg:dro}
    \begin{algorithmic}[1] 
        \State \textcolor{black}{Input: dataset $\ndataset = \{g_{\mech,t}^i,\,\hat{\msb}_t(\mech_{n,a},\cdot), t\in[T]\}_{i\in[\na]}$, Wasserstein radius $\epsilon$, stopping tolerance $\delta$}. 
        \State Initialize: $\hat{\psi} \in \Psi, \hat{v} \in \CV, \tilde{\Gamma}_k \leftarrow \emptyset, CV_k = \delta+1 \,\,\forall k\in[N]$.
        \While{$\max_{k\in[N]}CV_k \geq \delta$}
            \For{$k = 1:N$}
                \State Solve \eqref{eq:mcv} with $\hat{\psi},\hat{\bv}$, \textcolor{black}{updating} $\hat{\Theta}_k$, $CV_k$.
                \If {$CV_k$ > 0} \State $\tilde{\Gamma}_k \leftarrow \tilde{\Gamma}_k \cup \hat{\Theta}_k$ \EndIf
            \EndFor
            \State Solve \eqref{eq:finred} with $\tilde{\Gamma} = \cup_{k=1}^N \tilde{\Gamma}_k$, returning $\hat{\psi}, \hat{\bv}$.
        \EndWhile
        \State Output: $\delta$-optimal solution $\hat{\psi}$ of \eqref{eq:siprog}; thus, of \eqref{eq:dro}.
        \State $\hat{L}_{DR}(\mech_{n,a}) = \sup_{Q\sim B_{\epsilon}(\prod_{t=1}^T\hat{\msb}_t)}\CE_{\Theta \sim Q}\left[h(\hat{\psi},\Theta)\right] = \max_{\Theta\in\Gamma_{\epsilon}}\arg\min_r : \hu_s^i - \hu_t^i - \hlam_t^i g_{\mech,t}^i(\gamma_s^i) \leq \hlam_t^i\,r $
        \State $\Gamma_{\epsilon} = \{\Theta = \{\gamma_t^i,t\in[T]\}_{i\in[\na]} \in \Gamma : \sum_{i=1}^{\na}\sum_{t=1}^T \min_{k\in[N]}\|\gamma_t^i - \hat{\gamma}_{t,k}^i\| \leq \epsilon\}$
    \end{algorithmic}
    \end{algorithm}
    
    \begin{theorem}[Semi-Infinite Program Reformulation]
    \label{thm:sireform}
        Under Assumption \ref{as:convex}, the distributionally robust optimization~\eqref{eq:dro} is equivalent to the following semi-infinite program:
      \begin{align}
        \begin{split}
        \label{eq:siprog}
           & \min_{\psi\in\Psi,\bv\in\CV} \, \epsilon \cdot v_{N+1} + \frac{1}{N}\sum_{k=1}^N v_k \,\,\,
             s.t. \, \sup_{\Theta \in \Gamma} G_k(\psi,\bv, \Theta, \hat{\dataset}) \leq v_k \,\, \forall k\in [N]
             \\& G_k(\psi,\bv, \Theta, \hat{\dataset}) := h(\psi,\Theta) - v_{N+1}\sum_{i=1}^{\na}\sum_{t=1}^{T}\|\gamma_t^i - \hat{\gamma}_{t,k}^i \|_2
             \vspace{-0.5cm}
        \end{split}
        \end{align}
    \end{theorem}

    \begin{proof}
        Under Assumption~\ref{as:confun} and Specification~\ref{as:convex}, $\Gamma$ and $\Psi$ are compact. From Specification~\ref{as:convex} we have observed that $h(\psi,\Theta) \leq V$. Next, $h(\psi,\Theta)$ is uniformly Lipschitz continuous in $\psi$ and $\Theta$. Thus we can apply Corollary 3.8 of \cite{luo2017decomposition}.
    \end{proof}

    \subsubsection{Computing Solutions via Exchange Methods} \eqref{eq:siprog} is a semi-infinite program because $\Theta = \otimes_{t=1}^T \CA_{\mech,t} \subseteq \reals^{knT}$. This can be solved via exchange methods \cite{hettich1993semi}, \cite{dong2021wasserstein}, \cite{joachims2009cutting} as follows. 
    We first approximate it by a finite optimization, then iteratively solve this while appending constraints. Let $\tilde{\Gamma}_k = \{\Theta_{k,1,}\dots,\Theta_{k,J_k}\}$ be a collection of $J_k$ elements in $\Gamma$, i.e., each $\Theta_{k,j}, \, k\in[N], j\in[J_k],$ is a dataset $\{\gamma_{t,j}^i, t\in[T]\}_{i\in[\na]}$.  Then we can approximate \eqref{eq:siprog} by the following finite program:

    \begin{align}
    \begin{split}
    \label{eq:finred}
        &\min_{\psi\in\Psi,\bv\in\CV} \, \epsilon \cdot v_{N+1} + \frac{1}{N}\sum_{k=1}^N v_k \\
        s.t. \, &\max_{\Theta_{k,j} \in \tilde{\Gamma}_k}  G_k(\psi,\bv, \Theta_{k,j}, \hat{\dataset}) \leq v_k \,\, \forall j\in[J_k], k\in[N]
    \end{split}
    \end{align}
    We can iteratively refine the constraints in the finite program \eqref{eq:finred} by introducing the following maximum constraint violation problem:
    \begin{align}
        \begin{split}
        \label{eq:mcv}
            CV_k = \max_{\Theta \in \Gamma} G_k(\hat{\psi},\hat{\bv},\Theta,\ndataset) - \hat{v}_k 
        \end{split}
    \end{align}
    where $\hat{\bv} := \{\hat{v}_1, \dots,\hat{v}_{N+1}\}, \hat{\psi} := \{\hu_t^i,\hlam_t^i, t\in[T]\}_{i\in[\na]}$ are optimal solutions to \eqref{eq:finred} under $\tilde{\Gamma} := \cup_{k=1}^N \tilde{\Gamma}_k$. Supposing $CV_k > 0$, we let $\hat{\Theta}_{k} \in \Gamma$ be the argument attaining this maximum and append it to $\tilde{\Gamma}_k$ in \eqref{eq:finred}. Then we iterate, tightening the approximation for the infinite set of constraints in \eqref{eq:siprog} until $CV_k \leq \delta\,\, \forall k\in[N]$; by \cite{dong2021wasserstein} this termination yields a $\delta$-optimal solution of \eqref{eq:siprog}, that is, a solution of \eqref{eq:siprog} with constraint bounds relaxed by $\delta>0$. 

\begin{figure}[h]
    \centering
    \begin{tikzpicture}
        \node[draw, rectangle, minimum width=1cm, minimum height=1cm, align=center] (designer) at (-1.5,0) {\small Algorithm~\ref{alg:amd}\\ \scriptsize Lines 17-20};
        
    \node[draw, rectangle, minimum width=1cm, minimum height=1cm, align=center] (mas) at (5.5,0) {\small Algorithm~\ref{alg:dro}};
        
        \draw[->] (designer.east) -- (mas.west) node[midway, above, font=\small] {$\ndataset = \{g_{\mech,t}^i,\,\hat{\msb}_t(\mech_{n,a},\cdot), t\in[T]\}_{i\in[\na]}$};


        \draw[->] (mas.south) -- +(0,-0.5)  -|
        (designer.south) node[near start,below,font=\small] {$r_a = \hat{L}^{\epsilon}_{DR}(\mech_{n,a})$};

    \end{tikzpicture}
    \caption{\small Implementation of Algorithm~\ref{alg:dro}. Within Algorithm~\ref{alg:amd}, instead of solving the minimization in lines 18-20, we solve the distributionally robust optimization \eqref{eq:dro} by implementing Algorithm~\ref{alg:dro}.}
    \label{fig:alg_int}
\end{figure}

Algorithm~\ref{alg:dro} illustrates this iterative procedure, and can be characterized by the following convergence rate.
\begin{theorem}
\label{thm:alg2convg}
 Recall $T$ is the number of time indices in dataset $\CD$, $\na$ is the number of agents, and $\delta>0$ is the solution approximation tolerance. Algorithm~\ref{alg:dro} terminates within 
    \[\mathcal{O}\left( \left(\frac{1}{\delta} + 1\right)^{2T\na + 2}\right)\]
    algorithmic iterations.
\end{theorem}

\begin{proof}
    This follows from Theorem 2 of \cite{dong2021wasserstein}.
\end{proof}

In the numerical simulations in Section~\ref{sec:numsim}, we have empirically observed much quicker convergence than that guaranteed by Theorem~\ref{thm:alg2convg}. In Section~\ref{sec:numsim} we illustrate the rapid convergence of Algorithm~\ref{alg:dro} in a standard game-theoretic setting.

\paragraph{Interfacing with the Automated Mechanism Design Procedure} Algorithm~\ref{alg:dro} can be interfaced with Algorithm~\ref{alg:amd} as in Figure~\ref{fig:alg_int}. Instead of solving the minimization in lines 18-20 of Algorithm~\ref{alg:amd}, we solve the distributionally robust optimization \eqref{eq:dro} by implementing Algorithm~\ref{alg:dro}.  For this implementation the key point is that, under the assumptions of Theorem~\ref{thm:sireform}, iterating Algorithm~\ref{alg:amd} until $\hat{L}^{\epsilon}_{DR}(\mech) \leq c$ \textit{implies that $L(\mech) \leq c$}, for any $c\geq 0$. \textit{Thus, we can achieve mechanism design for the full mixed-strategies $\msb_t$ by only observing the partially specified mixed-strategies $\hat{\msb}_t$, as formalized in Lemma~\ref{lem:drogr}.}

\paragraph{Robustness to Constraint Misspecification}
\textcolor{black}{The same distributionally robust optimization machinery used to hedge against distributional uncertainty over observed
strategies also provides robustness to bounded perturbations in the constraint functions of
Assumption~1.  In practice, this allows the designer to operate effectively even when data are
collected under mechanisms not fully controlled or exactly satisfying the assumed structure of these constraints. For brevity, we do not include details of this reformulation, but this can be accomplished by a straightforward extension of the ambiguity set. }

\paragraph{Summary} \textcolor{black}{In this section we provided a distributionally robust optimization (DRO) formulation which allows for Pareto gap minimization from \textit{finite-sample} empirical strategy distributions. We provided a reformulation of this DRO as a semi-infinite program, and presented a finite-exchange method which allows us to achieve approximate solutions with arbitrary accuracy. As a whole, this methodology fits into our mechanism design framework by replacing the Pareto gap \eqref{eq:ineq_aug} evaluation by Algorithm~\ref{alg:dro}, enabling Pareto gap minimization in practical limited-sample settings.}
\section{Numerical Experiments}

\textcolor{black}{Here we provide two numerical experiments. The first example illustrates that  Algorithm~\ref{alg:amd}  iteratively decreases the Pareto gap and increases the social welfare, within a wireless spectrum sharing market. The second example illustrates that Algorithm~\ref{alg:amd} converges to the Pareto gap global minima, and also illustrates the rapid convergence of Algorithm~\ref{alg:dro}, in a canonical river pollution game setup.}  
\label{sec:num_wireless}
\noindent
\subsection{Example 1: Wireless Spectrum Sharing}
{\color{black} Here we illustrate the performance of Algorithm~\ref{alg:amd} in minimizing the Pareto gap and maximizing the social \textcolor{black}{welfare}, in a setting that explicitly exhibits the \emph{information asymmetry} between agents and designer that necessitates our revealed preference approach. In particular, we continue with  the example formulated in  Section~\ref{sec:netex}. Open-source code for this example is at  \texttt{https://github.com/LukeSnow0/Mechanism-Design}.

We consider a wireless spectrum--sharing market consisting of $M$ communication firms (agents), each transmitting at a power level (or bandwidth) $a_i \in [0,A_{\max}] \subset \reals_+$. 
Each firm $i$ has a private valuation parameter $v_i>0$ (unknown to the designer) and experiences interference from other firms $j$ through a random channel--gain matrix $\alpha_{ji} \ge 0$, $i,j\in[\na]$. 
The regulator (designer) enforces a pricing mechanism parameterized by $\mech=(\mech^0,\mech^1) \in \reals^2$ which determines the congestion--dependent cost faced by each agent.

\paragraph{Agent model.}
Given a mechanism $\mech$, each agent $i$ selects a mixed strategy $\mu_i \in \Delta([0,A_{\max}])$ over feasible power levels to maximize its expected payoff with respect to the following utility function
\begin{equation}
\label{eq:ag_utilities}
    f^i_{\mech}(a_i,a_{-i}) 
    = v_i \log(1+\tfrac{a_i}{\sigma^2 + \sum_{j\neq i}\alpha_{ji}a_j}) 
      - (\mech^0+\mech^1\!\sum_{j}a_j)\, a_i,
\end{equation}
where $\sigma^2$ is the noise variance parameter. Agents possess full knowledge of their own utility parameters and can estimate the aggregate impact of others’ actions (e.g., interference) through local observation or repeated interaction. This local information suffices for decentralized adaptation toward equilibrium: each agent updates its mixed strategy according to a \emph{logit quantal response equilibrium} (logit–QRE) dynamic \citep{mckelvey1995quantal}, whereby strategies with higher expected payoffs are chosen with higher probability. Over time, such stochastic best-response dynamics converge to fixed points representing equilibrium behavior in many classes of potential and congestion games \citep{hofbauer2002global,blume1993statistical}. The designer, however, observes only these equilibrium distributions ${\mu_i(\mech)}$, not the underlying utilities that generate them.

\begin{figure}[h!]
\begin{center}
    \includegraphics[width=0.6\linewidth]{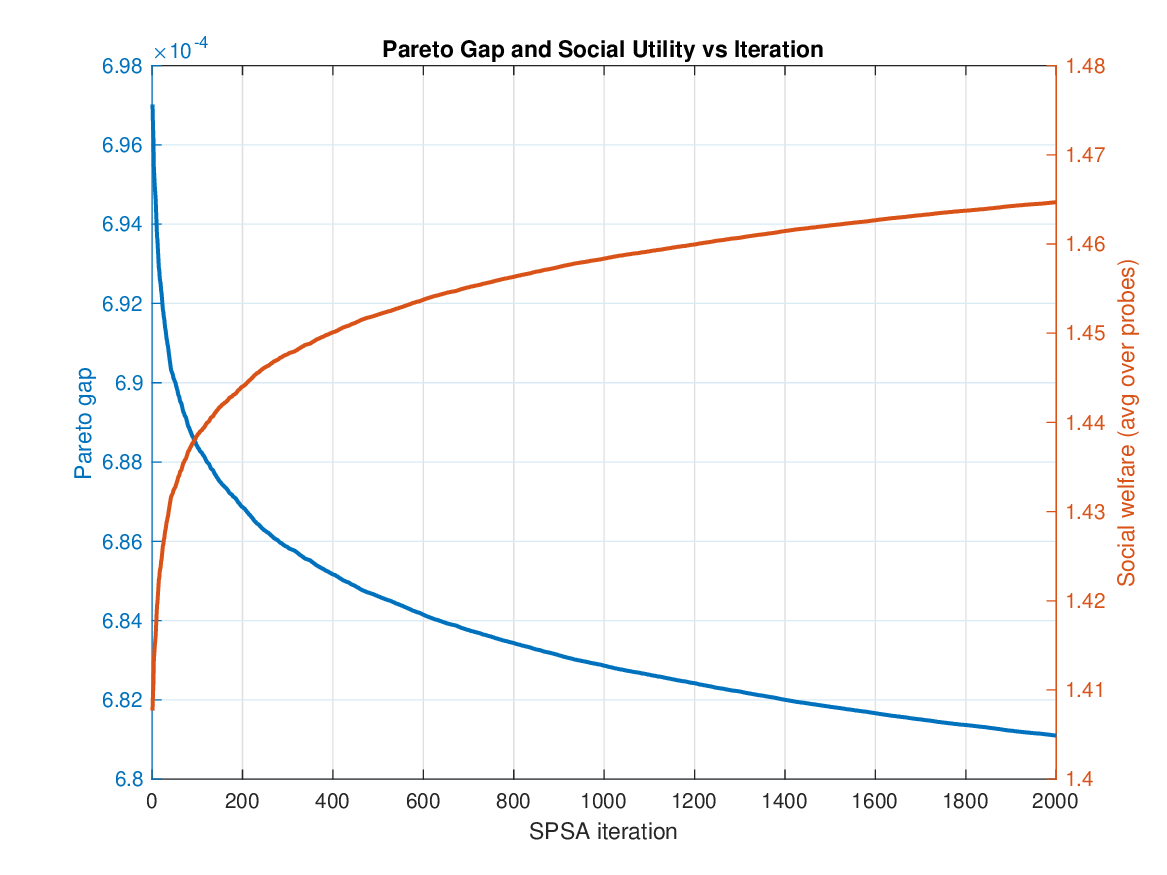}
\caption{\small \textcolor{black}{Evolution of the Pareto gap $L(\mech)$ \eqref{eq:ineq_aug} (left axis) and social \textcolor{black}{welfare} (right axis) with SPSA iterations. We observe monotonic decrease of the Pareto gap over SPSA iterations, and corresponding monotonic increase in social welfare \eqref{eq:soc_util}. Thus, we demonstrate the ability of the designer to adaptively improve social welfare by just observing sequential equilibrium strategies $\{\mu_t(\mech)\}_{t=1}^T$ and evaluating the Pareto gap \eqref{eq:ineq_aug}, \textit{without observation of utilities} \eqref{eq:ag_utilities}. This experiment is consistent with Theorem~\ref{thm:convg}, which guarantees asymptotic convergence in probability to $\arg\min L(\mech)$; however, here we display only the finite-iteration welfare improvement, demonstrating practical usefulness of Algorithm~\ref{alg:amd}.}}\label{fig:netspsa}
\end{center}
\end{figure}

\paragraph{Designer model.}
The regulator aims to adaptively tune $\mech$ so that the induced mixed--strategy Nash equilibria are \emph{socially optimal}. 
At each iteration the designer:
\begin{enumerate}
    \item generates a set of probing constraint functions $g_{i,t}(x)=\sqrt{\max(x-a_t\beta_i,0)}-1$ satisfying Assumption~\ref{as:confun};
    \item observes the equilibrium strategies $\{\mu_t(\mech)\}_{t=1}^T$ corresponding to these probes;
    \item evaluates the \emph{Pareto gap} \eqref{eq:ineq_aug} via the linear program
    \[
        \min_{r,\{u_{i,t},\lambda_{i,t}\}}\;\; r 
        \quad\text{s.t.}\quad 
        u_{i,s}-u_{i,t}-\lambda_{i,t}\, g_{i,t}(\mu_s)\le  \lambda_{i,t} \ r,\;
        \lambda_{i,t}\ge \underline{\lambda},
    \] 
    \item updates the mechanism parameter via a Simultaneous Perturbation Stochastic Approximation (SPSA) (Algorithm~\ref{alg:amd}) step
    \[
        \mech_{k+1} = \mech_k - a_k\, \widehat{\nabla}_\mech L(\mech_k), 
        \qquad
        \widehat{\nabla}_\mech L(\mech_k)
        = \frac{L(\mech_k^+)-L(\mech_k^-)}{2c_k}\,\Delta_k.
    \]
\end{enumerate}

\paragraph{Results.}
The simulation was implemented in MATLAB for $M=4$ agents, $T=10$ probing constraint sets, and $N_a=25$ discrete action levels.  
Agents' valuations $v_i$ and interference coefficients $\alpha_{ji}$ were drawn i.i.d.\ from $\mathcal{N}^+(0.8,0.4^2)$ and $\mathcal{U}(0,0.3)$, respectively.  
The logit temperature $\tau=6$ produced smooth mixed strategies.  
Over $1000$ SPSA iterations, the Pareto gap $L(\mech)$ decreased monotonically, demonstrating social welfare improvement.
Figure~\ref{fig:netspsa} shows the evolution of both the Pareto gap and the corresponding social \textcolor{black}{welfare}
\[
    \textstyle \sum_i \mathbb{E}_{a\sim\mu^i}[f^i_{\mech}(a,a_{-i})],
\]
averaged across probes. 
As the gap diminishes, the social \textcolor{black}{welfare} increases, confirming the alignment between the revealed--preference loss and utilitarian welfare.

\noindent
This experiment demonstrates that the proposed RL--based mechanism design procedure successfully improves pricing rules using \textit{only equilibrium strategy observations}, validating its applicability to real-world markets with hidden utilities.}

\subsection{Example 2: River Pollution Game}
\label{sec:numsim}

We implement the classic 'river pollution' game, see \cite{krawczyk2006nira}.  Three agents $i=1,2,3$ are located along a river. Each agent $i$ is engaged in an economic activity that produces pollution as a level $x_i$, and the players must meet environmental conditions set by a local authority. Pollutants are expelled into a river, where they disperse. There are two monitoring stations $l=1,2$ along the river, at which the local authority has set the maximum pollutant concentration levels. The revenue for agent $i$ is 
\[R_i(x) = [d_1  - d_2 \sqrt{x_1+x_2+x_3}]x_i\] where $d_1,d_2$ are parameters that determine an inverse demand law \cite{krawczyk2006nira}. Agent $i$ has expenditure 
\[F_i(x) = (c_{1i} + c_{2i} x_i)x_i\]
and thus has total profit 
\begin{equation}
\label{eq:payoff}
    f^i(x) = [d_1  - d_2\sqrt{x_1 + x_2 + x_3} - c_{1i} - c_{2i} x_i]x_i
\end{equation}

The local authority imposes the constraint on total pollution levels at each monitoring site as 
\begin{equation}
\label{eq:nlcon}
    q_l(x_1,x_2,x_3) = \sum_{i=1}^3\delta_{il}\, e_i\, x_i \leq 100, \, \, l=1,2
\end{equation}
where parameters $\delta_{il}$ are transportation-decay coefficients from player $i$ to location $l$, and $e_i$ is the emission coefficient of player $i$. 

\paragraph{Mechanism Design Implementation} We consider the perspective of the designer (local authority), who \textcolor{black}{\textit{does not know} or observe} the agent utility functions \eqref{eq:payoff}, but can enforce the constraints \eqref{eq:nlcon} and can modify the parameters $d_2,c_{1i},c_{2i}$. The parameter vector $[d_2, c_{11},c_{12},c_{13},c_{21},c_{22},c_{23}] \in \reals^7$ can be treated as a mechanism, since it modifies the agent utility functions and can be controlled by the designer. Also, the parameter vector $[e_1,e_2,e_3] \in \reals^3$ can be used to modify the constraint functions, and "probe" the system, as discussed in Section~\ref{sec:algder}. Note that  the utility functions \eqref{eq:payoff} and constraints \eqref{eq:nlcon} satisfy Assumptions \ref{as:pay}, \ref{as:con} and \ref{as:pspace}. 

We implement Algorithm~\ref{alg:amd} in Matlab. The non-cooperative game-enactment in steps 14-15 is achieved by the NIRA-3 Matlab package \cite{krawczyk2006nira}, which employs the relaxation algorithm and Nikaido-Isoda function to compute Nash equilibrium solutions. We  used the  Matlab linear program solver \texttt{linprog} to solve the optimization in lines 18-19. We set the lower bound on $L(\mech)$ to 0, so that the Nash equilibria are consistent with social optimality if and only if $L(\mech)=0$. We take the constraint functions \eqref{eq:nlcon} as generated by random parameters $\gamma_t = [e_1,e_2,e_3] \sim U[0,1]^3$ for each $t\in[T]$, i.e., each $e_i$ is uniformly distributed on $[0,1]$. The remainder of the algorithm is implemented with ease, taking $p=7,c=a=0.5,\eta=1/4, q=0.001, T=10$.  

We chose the  mechanism parameter $\mech = [d_2, c_{11},c_{12},c_{1,3},c_{21},c_{22},c_{23}]$ and $\mechspace = [0,1]^7$, that is, each element of $\mech$ is restricted to the unit interval $[0,1]$. We initialize $\mech_1$ as a $7$-dimensional uniform random variable in $[0,0.5]^7$, and iterate Algorithm~\ref{alg:amd} until $L(\mech)=0$. 

\begin{figure}[t]
    \begin{subfigure}{\textwidth}
      \includegraphics[width=\linewidth]{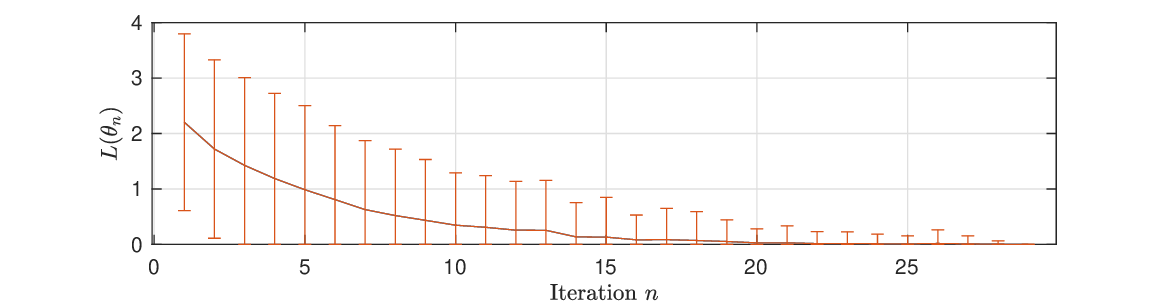}
    \end{subfigure}
    

    \caption{\small The Pareto gap $L(\mech_n)$, defined in \eqref{eq:ineq_aug}, quantifies the distance between empirical Nash equilibria and the social optima. Thus, we illustrate the effectiveness of the SPSA  algorithm, Algorithm~\ref{alg:amd}, in rapidly steering empirical Nash equilibria to the social optimum in this standard river pollution game.}
    \Description[description]{$L(\mech_n)$ quantifies the distance between empirical Nash equilibria and the social optima. Thus, we illustrate the effectiveness of the SPSA  algorithm, Algorithm~\ref{alg:amd}, in rapidly steering empirical Nash equilibria to the social optimum in this standard river pollution game.}
    \label{fig:SPSAcon}
\end{figure}

\paragraph{Results and Interpretation} Figure~\ref{fig:SPSAcon} displays the  value of $L(\mech_n)$ as a function of iteration $n$, where each iteration is averaged over 200 independent sample paths. The error bars represent one standard deviation in the 200 simulations. Recall that a mechanism $\mech^*$ such that $L(\mech^*) = 0$ induces social optimality \eqref{eq:md}, and also rank optimality \eqref{eq:rankopt} by Lemma~\ref{lem:rankopt}. It can be seen that in this example the algorithm converges rapidly, and achieves $L(\mech_n) = 0$ in less than 30 iterations in every case. Thus, while Theorem~\ref{thm:convg} gives asymptotic convergence guarantees, this implementation suggests that there are tighter finite-sample convergence bounds which should apply. It is an interesting point for future research to investigate the non-asymptotic performance of Algorithm~\ref{alg:amd}, see e.g., \cite{raginsky2017non}. 

By Theorem~\ref{thm:MA_Af} and Lemma~\ref{lem:rankopt}, achieving $L(\mech_n) = 0$ induces both social optimality and rank optimality, and thus achieves the mechanism design goal \eqref{eq:md}. However, supposing we are only able to obtain some $L(\mech_n) = c>0$, we may quantify its \textit{proximity} to social optimality through the relaxed revealed preference metrics of Section~\ref{sec:robustrp}.


 \begin{figure}[h!]
 \center
  \begin{subfigure}{\textwidth}
    \begin{center}
      \includegraphics[scale=0.5]{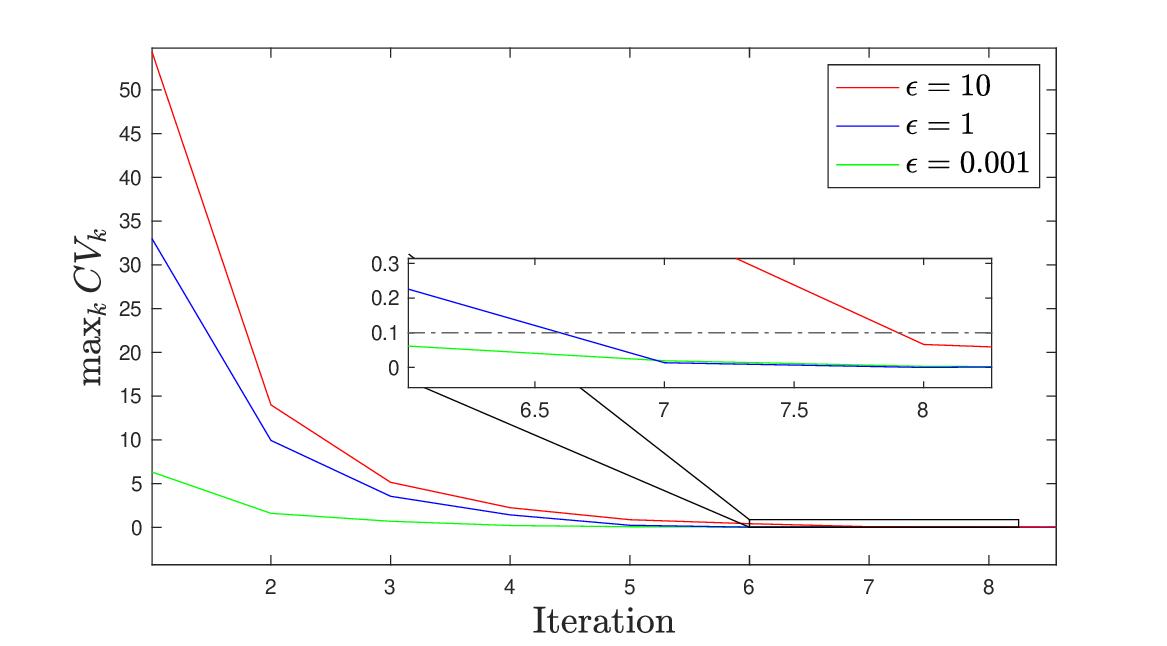}
    \end{center}
    \end{subfigure}
    
\vspace{-0.2cm}
    \caption{\small Convergence of Algorithm~\ref{alg:dro} for varying Wasserstein radii $\epsilon$, averaged over 200 Monte-Carlo simulations with randomly initialized parameters \eqref{eq:alg2int}. It can be observed that in each case the algorithm converges rapidly, with convergence rate inversely proportional to the Wasserstein proximity $\epsilon$.}
    \label{fig:Algcon}
\end{figure}

\subsubsection*{Distributionally Robust RL for Mechanism Design} Here we show that Algorithm~\ref{alg:dro} converges rapidly to a $\delta$-optimal solution of \eqref{eq:dro}. Despite the apparent complexity of the semi-infinite program reformulation \eqref{eq:dro}, in the numerical examples we considered, the finite program \eqref{eq:finred} yielded an accurate approximation of its solution within 8  iterations. 

The input to Algorithm~\ref{alg:dro} is the dataset $\ndataset = \{g_{\mech,t}^i,\,\hat{\msb}_t(\mech_{n,a},\cdot), t\in[T]\}_{i\in[\na]}$. For the purposes of simulation we consider the following simplified linearly constrained Nash equilibria generation: 
    \begin{align}
    \begin{split}
    \label{eq:alg2int}
    &\alpha_t^i \sim \mathcal{U}(0.1,1.1)^2 \in \reals^2,\quad \gamma_t^i \in \reals^2, \, t\in \{1,\dots,5\},\, i\in [\na] \\& \msb_t^i \in \arg\max_{\msb^i}\CE_{\gamma^i \sim \msb^i}[f^i(\gamma^i, \msb_t^{-i})]\,\, s.t.\, \,\langle\alpha_t^{i},\gamma^i \rangle\leq 1 \\
    &\gamma_{t,k}^i \sim \msb_t^i \,\, \forall k\in[N] 
    \end{split}
    \end{align}
    where  $\boldsymbol{1} =[1, 1]^\prime$, $\max$ operates elementwise, and the utilities of the 3 agents are $f^1(\beta) = \beta(1) + \beta(2), \,f^2(\beta) = \beta(1) + \beta(2)^{1/4},\, f^3(\beta) = \beta(1)^{1/4} + \beta(2)$. Without loss of generality we let $\lambda_t^i > 1 \, \forall t\in[T],i\in[\na]$. 
    
    We initialize $\delta = 0.1$ in Algorithm~\ref{alg:dro} and run until $\max_k CV_k \leq \delta$ for varying Wasserstein radii $\epsilon = 0.001, 1, 10$, illustrating the convergence in Figure~\ref{fig:Algcon}. Each curve is the average of 200 Monte-Carlo simulations \textcolor{black}{using randomly initialized parameters \eqref{eq:alg2int}}, and in each case Algorithm~\ref{alg:dro} produces a $\delta$-optimal solution on average within 8 iterations for $\delta = 0.1$.

\textcolor{black}{This experiment demonstrates the efficacy of Algorithm~\ref{alg:amd}, in quickly finding the global minima of the Pareto gap $L(\mech)$, in this case corresponding to social optimality ($L(\mech)=0$), and in handling practical finite-sample empirical strategy distributions through the DRO procedure of Algorithm~\ref{alg:dro}.}
    
\section{Conclusion}
\label{sec:conclusion}
We discussed a novel reinforcement learning (RL) methodology for achieving the mechanism design objective when the agent utility functions are unknown \textcolor{black}{and unreported} to the designer. The paper comprised four main results: \\
(i) We  generalized the result of \cite{cherchye2011revealed} to incorporate nonlinear budget constraints and mixed strategies \textcolor{black}{(Theorem~\ref{thm:MA_Af})}. This gives necessary and sufficient conditions for observed mixed-strategy Nash equilibria behavior to be consistent with social (Pareto) optimality. \\
(ii) We constructed a loss function which quantifies the Pareto gap of a set of mixed strategies, i.e., their distance from social optimality. We proved that an RL SPSA algorithm will converge in probability to the set of global minima of this loss function. \textcolor{black}{This enables a principled test for achievability of social optimality: if social optimality is achievable for the game setup and observational equilibria, Algorithm~\ref{alg:amd} will achieve it; else, we can \textit{identify} this unachievability, while still converging to a mechanism which is closest to social optimality}. The key distinguishing feature of this RL approach is that we achieve mechanism design \textit{without observations or reports of agent utilities}; this is enabled by our IRL Pareto gap evaluation step. \\
(iii) We provided an equivalence between the Pareto gap loss function and several robust revealed preference metrics, allowing for algorithmic sub-optimality guarantees in the case when \textcolor{black}{social optimality is unachievable}. \\
(iv) We described a  distributionally robust RL procedure which achieves \textcolor{black}{the mechanism design objective} for the full mixed-strategies from only a finite number of i.i.d. samples.

We illustrated the performance of Algorithm~\ref{alg:amd} and Algorithm~\ref{alg:dro} numerically in standard non-cooperative game settings. \textcolor{black}{It would be an interesting future endeavor to modify this algorithmic approach to not only induce Pareto optimality in the empirical dataset, but to steer the dataset towards specific regions in the Pareto frontier.}


\bibliography{Bibliography}
\bibliographystyle{abbrv} 

\appendix

\section*{Appendix}
In this appendix we first provide details of rank optimality, and a result showing that social optimality implies rank optimality in the ordinal preference setting, in Section~\ref{sec:rankopt}. Finally we provide the proofs of all theorems in the main texts in Appendix~\ref{sec:proofs}.

\section{Design Objective. Rank Optimality}
\label{sec:rankopt}
Here we discuss the conceptual motivation for extending the notion of optimality from \textit{cardinal} social optimality, specified by \eqref{eq:rpso} through utility functions $f_{\mech}^i$, to \textit{ordinal} optimality, as studied in \cite{chakrabarty2014welfare}. We will show that this objective captures a powerful new quantification of optimality, and under reasonable conditions is \textit{subsumed under the condition of social optimality \eqref{eq:sodef}}. 


Each utility function $f_{\mech}^i$ corresponds to a partial ordering ("preference profile") $\succ_i$ over outcomes $o_{\mech}(\cdot) \in O$: for any two joint-actions $\ba_1,\ba_2 \in \CA$, outcome $o_{\mech}(\ba_1)$ is \textit{preferred} to outcome $o_{\mech}(\ba_2)$ by agent $i$, denoted $o_{\mech}(\ba_1
) \succ o_{\mech}(\ba_2)$, if and only if $f_{\mech}^i(\ba_1) > f_{\mech}^i(\ba_2)$. Now, using these preference orderings, we may construct the following notion of social welfare known as \textit{rank optimality}. Let $\succ = (\succ_1,\dots,\succ_{\na})$ be the collection of preference orderings specified by utility functions $f_{\mech}^i, \, i\in[\na]$. Then, for $k\in[D]$, we denote the \textit{$k$-rank} of an outcome $o\in O$ in $\succ$, $\rnk_k(o;\succ)$, as the number of agents having $o \in O$ in their top $k$ choices:
\[\rnk_k(o;\succ) = |\{j : \pos(\succ_j,o)\leq k\}|\]
where $\pos(\succ_j,o)$ denotes the rank of outcome $o$ in preference profile $\succ_j$. Then we may define the \textit{max-rank} metric:
\[\maxrnk_k(\succ) := \max_{o\in O} \rnk_k(o;\succ)\]
as the maximum possible $k$-$\rnk$ over all outcomes for the preference profile $\succ$. Thus, we may define a \textit{rank-optimal} strategy as follows.
\begin{definition}[Rank Optimal Strategy]
For a preference profile $\succ = (\succ_1,\dots,\succ_{\na})$ and mechanism $\mech$, a joint-action $\ba$ is rank-optimal if 
\[\rnk_k(o_{\mech}(\ba);\succ) = \max\rnk_k(\succ)\,\,\, \forall \, k\in[D]\]
Similarly, a mixed-strategy $\msb$ is rank-optimal if 
\[\CE_{\ba\sim\msb}[\rnk_k(o_{\mech}(\ba);\succ)]= \max\rnk_k(\succ)\,\,\, \forall \, k\in[D]\]
\end{definition}
This notion of rank-optimality is well-motivated as a social welfare design principle for mechanism design \cite{chakrabarty2014welfare}, and may replace the social optimality criterion \eqref{eq:sodef} in certain settings. One such setting is when the agents have \textit{ordinal} preferences, i.e., can only compare finitely many outcomes against one another, without quantifying their utility explicitly. In this case, the utility functions $\{f_{\mech}^i,i\in[\na]\}$ are constructed synthetically from the preference profile $\succ$, and are regarded as \textit{homogeneous}:
\begin{definition}[Homogeneous Utility Functions]
\label{def:homuf}
Let $|\CA| = D < \infty$, and denote $f_{\mech}^i(j), j\in[D]$ the utility of agent $i$'s $j$'th preferred outcome, under preference profile $\succ$. Utility functions $\{f_{\mech}^i\}_{i\in[\na]}$ are said to be homogeneous if for every $\mech \in \mechspace, i,i'\in[\na]$, we have $f_{\mech}^i(j) = f_{\mech}^{i'}(j)$.
\end{definition}
Utility functions are homogeneous if each agent assigns the same value to their $j$'th preference, for every $j\in[D]$. This arises out of straightforward construction of utility functions $f_{\mech}^i(\cdot)$ from a finite preference profile $\succ$. 

The setting of ordinal preference relations and homogeneous utility functions is an important practical setting, and can easily be recovered as a special case of the general game form \eqref{eq:gmixed}. In this special case, the mechanism design goal should be to find a mechanism inducing \textit{max-rank} strategies. We now show that this is an immediate consequence of our general social optimality objective \eqref{eq:sodef}.

\begin{lemma}[Social Optimality Induces Rank Optimality]
\label{lem:rankopt}
Assume a finite joint-action space $\CA$: $|\CA| = D< \infty$ and homogeneous utility functions $\{f_{\mech}^i(\cdot)\}_{i\in[\na]}$ inducing preference profile $\succ = (\succ_1,\dots,\succ_{\na})$. Assume the relation \eqref{eq:sodef} holds. Then, we have 
\begin{equation}
\label{eq:rankopt}
    \CE_{\ba\sim\msb}[\rnk_i(\ba;\succ)] = \max\rnk_i(\succ) \,\, \forall i \in[D]
\end{equation}
\end{lemma}
Thus, the goal of rank-optimality for the setting of finite action space and ordinal preferences is subsumed under our general mechanism design goal of social optimality \eqref{eq:md}.

\section{Proofs}
\label{sec:proofs}
\subsection{Proof of Lemma~\ref{lem:rankopt}}
By \eqref{eq:md} we have $\msb \in \arg\max_{\ms\in\Delta(\CA)}\sum_{i=1}^{\na}f_{\mech}^i(\ms) := \arg\max_{\ms\in\Delta(\CA)}\sum_{i=1}^{\na}\CE_{\ba\sim\ms}[f_{\mech}^i(\ba)]$. $\{f_{\mech}^i(\cdot)\}$ are homogeneous, so denote $U(j) = f_{\mech}^i(j)$, using the notation in Def.~\ref{def:homuf}. Also let $U(D+1)=0$ and $\rnk_0(o;\succ) = 0\,\forall o\in O$ (for notational convenience in what follows).
Now, we can make the following expansion
\begin{align*}
    \sum_{i=1}^{\na}f_{\mech}^i(\msb) &= \sum_{\ba\in\CA}\msb(\ba)\sum_{j=1}^D (\rnk_j(o_{\mech}(\ba); \succ) - \rnk_{j-1}(o_{\mech}(\ba); \succ))U(j)\\
    &= \sum_{\ba\in\CA} \sum_{j=1}^D \rnk_j(o_{\mech}(\ba); \succ)(U(j) - U(j+1)) \\
    & = \sum_{j=1}^D \CE_{\ba \sim \msb} (\rnk_{j}(o_{\mech}(\ba); \succ))(U(j) - U(j+1)) \\
    & = \max_{\ms\in\Delta(\CA)}\sum_{i=1}^{\na}f_{\mech}^i(\ms) = \max_{\ba\in\CA}\sum_{j=1}^D(\rnk_j(o_{\mech}(\ba);\succ)(U(i) - U(i+1)) \\
    & = \sum_{i=1}^D \max\rnk_j(\succ)( U(j) - U(j+1))
\end{align*}
Thus we can conclude that $\CE_{\ba \sim \msb} (\rnk_{j}(o_{\mech}(\ba); \succ)) = \max\rnk_j(\succ)\, \forall j\in[D]$, since 
\[\rnk_j(o_{\mech}(\ba);\succ) \leq \max\rnk_j(\succ) \,\forall \ba\in\CA\]

\subsection{{Proof of Theorem~\ref{thm:MA_Af}}}
\label{ap:pfa}
We show the equivalence between condition $i$ and $iii$. The equivalence to $ii$ then follows from \cite{forges2009afriat} Proposition 3.

$(i \Rightarrow iii)$: Assume \eqref{eq:rpso} holds. We have that $\msb_t$ solves:
\begin{equation}
    \label{eq:max}
        \msb_t \in \arg\max_{\ms \in \Delta(\CA_{\mech,t})} \sum_{i=1}^{\na} f^i(\ms) \, \, \forall t\in[T]
\end{equation}
where $f^i(\ms) = \int_{\reals^{\color{black} k}}f^i(x)d\ms^i(x)$ and $\CA_{\mech,t} = \{x = [x^1,\dots,x^{\na}]\in\reals^{k\na} : g_{\mech,t}^i(x^i) \leq 0\,\,\forall i\in[\na]\}$. Since $f^i(\cdot)$ is continuous for each $i$, there exists a point $\bar{x}_t = [\bar{x}_t^1,\dots,\bar{x}_t^{\na}]'\in \CA_{\mech,t}$ such that $\sum_{i=1}^{\na} f^i(\ms_t) = \sum_{i=1}^{\na} \int_{\reals^{\color{black} k}}f^i(x)d\ms(x) = \sum_{i=1}^{\na}f^i(\bar{x}_t)$. Now observe that the point $\bxt$ must solve the pure strategy optimization
\begin{equation}
    \label{eq:maxaug}
        \bxt \in \arg\max_{x = [x^1,\dots,x^{\na}]'} \sum_{i=1}^{\na}f^i(x) \, \, s.t. \,\,  g_{\mech,t}^i(x) \leq 0 \,\, \forall i \in [\na]
\end{equation}
since otherwise there would exist some $x'\in\CA_{\mech,t}$ such that $\sum_{i=1}^{\na}f^i(x') > \sum_{i=1}^{\na} f^i(\bxt) = \sum_{i=1}^{\na} f^i(\msb_t)$, which would contradict \eqref{eq:max} since $\sum_{i=1}^{\na} \int_{\CA_{\mech,t}}f^i(x)\delta(x-x')dx > \sum_{i=1}^{\na} \int_{\CA_{\mech,t}}f^i(x)\msb_t(x)$. Furthermore observe that since $g_{\mech,t}^i$ and $f^i$ are increasing we have $g_{\mech,t}^i(\bxti) = 0$.

Now by concavity, the functions $f^i(\cdot)$ are subdifferentiable, and thus the sum $\sum_{i=1}^{\na}f(\cdot)$ is subdifferentiable. Then, letting $\et^i$ be the Lagrange multiplier associated with budget constraint $g_{\mech,t}^i(\cdot)$, an optimal solution $\bar{x}_t$ to the maximization problem \eqref{eq:maxaug} must satisfy
    $\sum_{i=1}^{\na} \nabla f^i(\bar{x}_t) = \sum_{i=1}^{\na}\et^i \nabla g_{\mech,t}^i(\bar{x}_t)$
    for subgradient operator $\nabla$ and where $\nabla g_{\mech,t}^i(x)  =[\boldsymbol{0},\dots,\nabla g_{\mech,t}^i(x^i)',\dots,\boldsymbol{0}]'$ since $g_{\mech,t}^i$ only depends on those elements corresponding to $x^i$.
    Now let 
    \begin{align*}
        \begin{split}
            \lambda_t^i &:= \frac{\left(\sum_{i=1}^{\na}\et^i\|\nabla g_{\mech,t}^i(\bar{x}_t)\|_{\infty} + \sum_{j\neq i} \|\nabla f^j(\bar{x}_t)\|_{\infty} \right)}{\min_{k\in[n]}|(\nabla g_{\mech,t}^i(\bar{x}_t))_k|}\\
            \gam_t^i &:= \nabla g_{\mech,t}^i(\bar{x}_t), \quad \gam_t = \sum_{i=1}^{\na} \gam_t^i= [\nabla g_{\mech,t}^1(x^1)',\dots,\nabla g_{\mech,t}^{\na}(x^{\na})']'
        \end{split}
    \end{align*}

    and observe that for each $i \in [\na]$, 
    \begin{equation}
    \label{eq:ineq1}
        \sgcf \leq \lambda_t^i \gam_t
    \end{equation}
    Next, concavity of the functions $\cf^i$ implies for each $i$
    \begin{equation}
    \label{eq:ineq2}
        f^i(\bxs) - f^i(\bxt) \leq \sgcf'(\bxs - \bxt)
    \end{equation}
    Substituting inequality \eqref{eq:ineq1} into \eqref{eq:ineq2}, and setting $v_s^i = f^i(\bxs), \,v_t^i = f^i(\bxt)$, we obtain $ v_s^i - v_t^i \leq \lambda_t^i(\gam_t)' (\bxs - \bxt)$
    Thus the set $\{\gam_t, \bxt, t \in[T] \}$ satisfies the linear 'Generalized Axiom of Revealed Preference' (GARP) \cite{cherchye2011revealed} for each $i \in \na$. We will show that this implies the dataset $\{g_{\mech,t}^i,\resp, t \in [T]\}$ satisfies \textit{nonlinear} GARP \cite{forges2009afriat}. First, assume $g_{\mech,t}^i(\bxsi) < 0$. Since $g_{\mech,t}^i(\bxti) = 0$ we have $g_{\mech,t}^i(\bxsi) < g_{\mech,t}^i(\bxti)$, and thus $\sgcn'(\bxsi - \bxti) \leq 0$ by monotonicity of $g_{\mech,t}^i$. So $(\gam_t^i)' \bxsi < (\gam_t^i)' \bxti$. Then by inverting this implication, we obtain $(\gam_t^i)' \bxsi \geq (\gam_t^i)' \bxti \Rightarrow g_{\mech,t}^i(\bxsi) \geq 0$.
    
    Since the linear GARP is satisfied for all $i\in[\na]$ we have the relation
    \begin{equation*}
        \bxs R \bxt \Rightarrow \gam_t' \bxt= \sum_{i=1}^{\na}(\gam_t^i)' \bxti \leq \sum_{i=1}^{\na}(\gam_t^i)'\bxsi = \gam_t' \bxs \Rightarrow g_{\mech,t}^i(\bxs) \geq 0  
    \end{equation*}
    Now, by our choice of probing constraint functions in Assumption~\ref{as:confun} we have that for each $t\neq s$:
    \[\exists i\in[\na] \,\,s.t. \,\,\gam_t^i\bxti > \gam_t^i \bxsi \Rightarrow \gam_t^k\bxt^k > \gam_t^k \bxs^k \,\,\forall k \neq i\]
    and thus $\bxs R \bxt \Rightarrow g_{\mech,t}^i(\bxs) \geq 0$, so $\{g_{\mech,t}^i, \resp, t \in [T]\}$
    satisfies nonlinear GARP for all $i\in[\numconsts]$. So, by \cite{forges2009afriat} we have that for each $i \in [M]$, there exist $u_t^i \in \reals, \lambda_t^i > 0$ such that 
    \[ u_s^i - u_t^i - \lambda_t^i(\nlconst^i(\bar{x}_s^i)) \leq 0 \quad \forall t,s \in [T]\]

    Now since \eqref{eq:max} holds, it must be the case that $\msb \in \Delta(\otimes_{i=1}^{\na}\del g_s^i)$, where $\del g_s^i = \{x \in \reals^{\color{black} k} : g_s^i(x)=0\}$, i.e., $\msb$ has measurable support on the boundary of $g_s^i(\cdot)$; this holds since $f^i(\cdot)$ and $g_s^i$ are both monotone element-wise increasing for every $i\in[\na]$. Thus, $\bar{x}_s^i \in \del g_s^i(\cdot) \,\, \forall i\in[\na],s\in[T]$. Now, by the structure of $g_s^i$ imposed by Assumption~\ref{as:confun}, we have that $g_{\mech,t}^i(x) = g_{\mech,t}^i(y) \, \forall x,y \in \del g_s^i(\cdot)$. Then, taking $\bar{\gamma}_t = [\bar{\gamma}_t^1,\dots,\bar{\gamma}_t^{\na}]$, where $\bar{\gamma}_s^i \in \del g_s^i = g_s^{i -1}(g_s^i(\msb))$ gives $g_{\mech,t}^i(\bar{\gamma}_s^i) = g_{\mech,t}^i(\bar{x}_s^i) \,\, \forall t,s\in[T]$, implying that \eqref{eq:ineq} holds.

    (iii $\Rightarrow$ i): Assume \eqref{eq:ineq} holds, take $\gamma$ such that $g_{\mech,t}^i(\gamma) \leq 0 \ \forall i \in [M]$ and define
    \begin{equation}
    \label{eq:Usumdef}
        \sum_{i=1}^{\na} U^i(\gamma) = \sum_{i=1}^{\na} \min_{s \in [T]} \left[u_s^i + \lambda_s^i[g_s^i(\gamma)] \right] 
    \end{equation}
    Now observe by the inequalities \eqref{eq:ineq}, and since $g_{\mech,t}^i(\bxsi) = g_{\mech,t}^i(\bar{\gamma}_s^i)\,\forall i\in[\na]$, we have $\sum_{i=1}^{\na} U^i(\bxt) = \sum_{i=1}^{\na} u_t^i $. So, considering $\gamma$ s.t. $g_{\mech,t}(\gamma) \leq 0 \, \, \forall i\in[\na]$, we have that
    \begin{equation*}
        \sum_{i=1}^{\na} U^i(\gamma) \leq \sum_{i=1}^{\na} \left[u_t^i + \lambda_t^i g_{\mech,t}^i(\gamma) \right] \leq \sum_{i=1}^{\na}U^i(\bxt).
    \end{equation*}
    Thus $\bxt$ is a solution to \eqref{eq:maxaug} with aggregate utility \eqref{eq:Usumdef}, for all $t\in[T]$. This implies that $\ms_t$ satisfies \eqref{eq:max} for all $t\in[T]$, with aggregate utility $\sum_{i=1}^{\na} f^i(\cdot) = \sum_{i=1}^{\na} U^i(\cdot)$ \eqref{eq:Usumdef}, since $\sum_{i=1}^{\na}f^i(\ms_t) = \sum_{i=1}^{\na}f^i(\bxt)$.

\subsection{Proof of Theorem~\ref{thm:cineq}}
\label{sec:lempf}

Suppose we have $\CE[\hat{L}(\mech)] =c$, where $\hat{L}(\mech)$ is defined in \eqref{eq:ineq_aug2} and $\CE[\hat{L}(\mech)]$ is given by 
\[\CE[\hat{L}(\mech)] = \int_{\otimes_{t=1}^T\Delta(\CA_{\mech,t})}\left[\arg\min_{r\geq 0}: u_s^i - u_t^i - \lambda_t^i g_{\mech,t}^i(\hat{\msb}_s) \leq \lambda_t r\right]d\PR(\otimes_{t=1}^T\hat{\msb}_t)\]
Then obviously there exists some set $\{u_t^i \in \reals,\lambda_t^i > \alpha\}$ and empirical strategy $\hat{\msb}_t(\mech)$ such that 
\[u_s^i - u_t^i - \lambda_t^i g_{\mech,t}^i(\hat{\gamma}_s^i) \leq \lambda_t^i c \, \, \forall t,s \in [T],i \in [\na]\]
Also, with 
\begin{align*}
&\hat{\gamma}_t^i \in g_{\mech,t}^{i^{-1}}(g_{\mech,t}^i(\hat{\msb}_t(\mech))) = g_{\mech,t}^{i^{-1}}\left(\frac{1}{N}\sum_{k=1}^N g_{\mech,t}^i(\gamma_{t,n,a}^k(\mech))\right), \\ &\bar{\gamma}_t^i \in g_{\mech,t}^{i^{-1}}(g_{\mech,t}^i(\msb_t(\mech))) = g_{\mech,t}^{i^{-1}}\left(\int_{\CA_{\mech,t}}g_{\mech,t}^i(x)\msb_t(\mech,x)dx\right)
\end{align*}
 by Hoeffding's inequality we can produce the concentration bound 
\[\mathbb{P}(|g_{\mech,t}^i(\hat{\gamma}_s^i) - g_{\mech,t}^i(\bar{\gamma}_s^i)| >\epsilon) \leq 2\exp\left(\frac{-2\epsilon^2N}{G^2} \right)\]
where $G := |\inf_{i,t} g_{\mech,t}^i(\boldsymbol{0})|$ is the maximum range over all $\{g_{\mech,t}^i(\cdot)\}$ since $\CA_{\mech,t}^i \subset \reals^k_+$ (non-negative orthant) and $x\in\CA_{\mech,t}^i \Rightarrow g_{\mech,t}^i(x)\leq 0$. Now, we have the relation 
\begin{align*}
    &\biggl\{|g_{\mech,t}^i(\hat{\gamma}_s^i) -  g_{\mech,t}^i(\bar{\gamma}_s^i)| < \epsilon \, \,\forall t,s\in[T],i\in[\na]\biggr\} \\& \subseteq \biggl\{\exists \{u_s^i \in \reals,\lambda_s^i > \alpha\} : u_s^i - u_t^i - \lambda_t^ig_{\mech,t}^i(\bar{\gamma}_s^i) \leq \lambda_t^i(c + \epsilon) \, \, \forall t,s\in[T],i\in[\na]\biggr\}
\end{align*}
and thus 
\begin{align*}
&\mathbb{P}\left(L(\mech) \leq c + \epsilon  \right) = \mathbb{P}\left(\exists \{u_s^i \in \reals,\lambda_s^i > \alpha\} : u_s^i - u_t^i - \lambda_t^i g_{\mech,t}^i(\bar{\gamma}_s^i) \leq \lambda_t^i(c + \epsilon) \, \, \forall t,s\in[T],i\in[\na] \right) \\
&\geq \mathbb{P}\left(|g_{\mech,t}^i(\hat{\gamma}_s^i) - g_{\mech,t}^i(\bar{\gamma}_s^i) | < \epsilon \,\, \forall t,s\in[T], i\in [\na] \right) \geq \prod_{t=1}^{T}\prod_{i=1}^{\na}\left(\min\biggl\{1-2\exp\left(\frac{-2\epsilon^2N}{G^2}\right),0 \biggr\} \right)^T 
\end{align*}

\subsection{Proof of Theorem~\ref{thm:convg}}
\label{sec:con_pf}
We consider the stochastic loss function
\begin{align}
\begin{split}
\label{eq:lp}
    &\hat{L}(\mech) = \arg\min_{r\in\reals} : \exists \{u_t^i \in \reals,\lambda_t^i > 
    \alpha, t\in[T], i\in[\na]\} :\\ &\quad u_s^i - u_t^i - \lambda_t^i\nlconst^i(\hat{\msb}_s(\mech)) \leq \lambda_t^i r \quad \forall t,s\in[T],i\in[\na] 
\end{split}
\end{align}
where $\msb_t(\mech)$ is a mixed-strategy Nash equilibrium solution for the game with mechanism $o_{\mech}$ and feasible set $\CA_{\mech,t}^i = \{a: g_{\mech,t}^i(a) \leq 0\}$. 

We show that $\hat{L}(\mech)$ and Algorithm~\ref{alg:amd} iterates satisfy assumptions H1-H8 of \cite{maryak2001global}. Then, we apply Theorem~\ref{thm:cineq} to obtain the result.

\textbf{H1}: $\Delta_n$ is distributed according to the $p$-dimensional Rademacher distribution, and thus the conditions are satisfied. 

\textbf{H2}: We have 
\begin{equation}
\label{eq:elh}
    \CE[\hat{L}(\mech)] = \int_{\otimes_{t=1}^T\Delta(\CA_{\mech,t})}\left[\arg\min_{r\geq 0}: u_s^i - u_t^i - \lambda_t^i g_{\mech,t}^i(\hat{\msb}_s) \leq \lambda_t r\right]d\PR(\otimes_{t=1}^T\hat{\msb}_t)
\end{equation}
Letting $\epsilon_{\mech}$ be a random variable such that $\hat{L}(\mech) = \CE[\hat{L}(\mech)] + \epsilon_{\mech}$ (this always exists as long as $\CE[\hat{L}(\mech)]<\infty)$). Then $\CE[\epsilon_{\mech}] = 0$ by construction. So, this condition is satisfied.

\textbf{H3}: 
To prove this condition we consider differentiability with respect to $L
(\mech)$; this then immediately transfers to $\CE[\hat{L}(\mech)]$ \eqref{eq:elh} by linearity.
Let $\mech = [\mech_1,\dots,\mech_p]$. By Fa\'a di Bruni's formula, the third derivative of $L(\mech)$ with respect to $\mech_k, k\in[p]$ can be decomposed as the combinatorial form:
\begin{align}
\begin{split}
\label{eq:thdriv}
    &\frac{\del^3 L(\mech)}{\del\mech_k^3} =\sum_{t,i}\sum_{\pi\in\Pi}\frac{\del^{|\pi|}L(g_{\mech,t}^i(\bar{\gamma}_s^i(\mech)))}{\del g_{\mech,t}^i(\bar{\gamma}_s^i(\mech)) ^{|\pi|}}\cdot \prod_{B\in\pi}\frac{\del^{|B|}g_{\mech,t}^i(\bar{\gamma}_s^i(\mech),))}{\del \mech_k^{|B|}}
\end{split}
\end{align}
where $\pi$ iterates through the set of partitions $\Pi$ of the set $\{1,2,3\}$, $B\in\pi$ is an element of partition $\pi$, and $|\pi|, |B|$ indicate the cardinality of said partitions and elements, respectively. We need to show that \eqref{eq:thdriv} exists and is continuous for all $k\in [p]$. 

Now, for notational simplicity re-write $g_{\mech,t}^i(\bar{\gamma}_s(\mech)))$ as $g_{t,s,i}$ and consider the term
$\frac{\del^{n}L(g_{t,s,i})}{\del g_{t,s,i}^{n}}$ for some $n\in[3]$. This expresses the $n$-th order derivative of the solution to the linear program \eqref{eq:lp} with respect to real number $g_{t,s,i}$. Fortunately we can obtain a useful bound on this value by the methodology in \cite{freund2009postoptimal}. Specifically, we can re-write the linear program \eqref{eq:lp} as 
\begin{align}
\begin{split}
\label{eq:LP}
    &\textrm{maximize } L(g_{t,s,i}) = c\cdot x \, \,s.t.\,\, Ax = 0, \, \, x \geq 0
\end{split}
\end{align}
where $x = [c,u_1^1,\dots,u_T^M,\lambda_1^1,\dots,\lambda_T^M]^T \in \reals^{2TM+1}, c = [-1,0,\dots,0] \in \reals^{2TM+1}$ and $A$ is a matrix of coefficients that produces the equalities
$u_s^i - u_t^i - \lambda_t^i g_{t,s,i} - r = 0 \, \, \,\,\forall s,t,i.$

Then, letting $\beta \subset [2TM+1]$, denote $A_{\beta}$ as the submatrix of $A$ with columns corresponding to elements of $\beta$, such that $A_{\beta}$ is nonsingular. Form $\hat{x}_{\beta}$ be the basic primal solution to \eqref{eq:LP} with the $i'th$ element of $\hat{x}_{\beta}$ equal to zero for $i\in[2TM+1]\backslash\beta$. Also let $G_{\beta}$ be the matrix $A_{\beta}$ with all elements apart from $g_{t,s,i}$ replaced by zero. Then by Theorem 1 of \cite{freund2009postoptimal}, we have that for $n\in[3]$, $L(g_{t,s,i})$ is $n$-times differentiable, with derivatives given by
\[\frac{\del^n L(g_{t,s,i})}{\del g_{t,s,i}^n} = (n!)c(-A_{\beta}^{-1}G_{\beta})^n\hat{x}_{\beta}\]
and for $g\in\reals$ sufficiently close to $g_{t,s,i}$,
\[\frac{\del^n L(g)}{\del g^n} = \sum_{i=n}^{\infty}\frac{i!}{(i-n)!}c(g-g_{t.s,i})^{(i-n)}(-B^{-1}G_{\beta})^i\hat{x}_{\beta}\]
This reveals that the $L(g_{t,s,i})$ is $n$-times continuously differentiable for $n\in[3]$. 

Now consider the term 
$ \frac{\del^{n}g_{\mech,t}^i(\bar{\gamma}_s^i(\mech))}{\del \mech_k^{n}}$ for $n\in[3]$. This term can again be decomposed using Fa\'a di Bruno's formula (generalization of the chain rule for higher order derivatives) as  \[\frac{\del^{n}g_{\mech,t}^i(\bar{\gamma}_s^i(\mech)))}{\del \mech_k^{n}}=\sum_{\pi\in\Pi}\sum_{j}\frac{\del^{|\pi|}g_{\mech,t}^i([\bar{\gamma}_s^i(\mech)]_j)}{\del [\bar{\gamma}_s^i(\mech)]_j^{|\pi|}}\cdot\prod_{B\in\pi}\frac{\del^{|B|}[\bar{\gamma}_s^i(\mech)]_j}{\del \mech_k^{|B|}}\] 
with $\Pi$ being the partitions of $[n]$ here and $[\bar{\gamma}_s^i(\mech)]_j$ the $j$'th element of vector $\bar{\gamma}_s^i(\mech)$. The constraint function derivatives are bounded by assumption~\ref{as:con}. The term 
$\frac{\del^{|B|}[\bar{\gamma}_s^i(\mech)]_j}{\del \mech_k^{|B|}} $ can be expressed as $\frac{\del^{|B|}[\langle g_s^{i^{-1}}(g_s^i(\ms_s^i(\mech)))\rangle_{c}]_j}{\del \mech_k^{|B|}} $ where $\langle g_s^{i^{-1}}(g_s^i(\msb_s(\mech)))\rangle_{c}$ is the unique value $\langle\bar{\gamma}_s^i(\mech)\rangle_{c} := \int_{g_s^{i^{-1}}(g_s^i(\msb_s(\mech)))}xdx$ representing the 'center' point of inverse set $g_s^{i^{-1}}(g_s^i(\ms_s^i(\mech))) \subset \reals^k$ so that the map $\langle g_s^{i^{-1}}(\cdot)\rangle_{c}$ is invertible.  Then, since the map $g_s^i$ is continuous with bounded derivatives by Assumption~\ref{as:con}, we can expand $\frac{\del^{|B|}[\langle g_s^{i^{-1}}(g_s^i(\ms_s^i(\mech)))\rangle_{c}]_j}{\del \mech_k^{|B|}}$ by the chain rule and consider only the term $\frac{\del^{|B|}g_s^i(\ms_s^i(\mech))}{\del \mech_k^{|B|}} = \frac{\del^{|B|}\int_{\CA_s^i}g_s^i(x)\ms_s^i(\mech,x)dx}{\del \mech_k^{|B|}}$. Now $\ms_s^i(\mech) \in \arg\max_{\mu\in\Delta(\CA_s^i)} f_{\mech}^i(\mu | \ms_s^{-i}(\mech)), \forall i$ by definition, and for each $i\in[\na]$, the utility obtained under parameter $\mech$ and joint mixed-strategy $\msb_s(\mech) = \otimes_{i=1}^{\na}\ms_s^i(\mech)$ is \[f_{\mech}^i(\msb_s) = \int_{\CA_s}f_{\mech}^i(x)\msb_s(x)dx = \int_{\CA_s^{-i}}\int_{\CA_s^i}f_{\mech}^i(x)\ms_s^i(x^i)\ms_s^{-i}(x^{-i})dx\] where $\CA_s^{-i} = \otimes_{j\neq i}\CA_s^j, x = [x^1,\dots,x^{\na}], \ \in \reals^{k\na}, \ms_s^{-i} = \otimes_{j \neq i} \ms_s^j, x^{-1} = \otimes_{j \neq i} x^j$. Now, by continuity of $f_{\mech}^i$, there must exist a point $\bxsi \in \CA_s^i$ such that \[f_{\mech}^i(\msb_s) = \int_{\CA_s^{-i}}\int_{\CA_s^i}f_{\mech}^i(x)\delta(x-\bxsi)\ms_s^{-i}(x^{-i})dx\] and furthermore \[\bxsi \in \arg\max_{\gamma\in\CA_s^i}\int_{\CA_s^{-i}}\int_{\CA_s^i}f_{\mech}^i(x)\delta(x-\gamma)\ms_s^{-i}(x^{-i})dx\] since otherwise the mixed strategy $\ms_s^i$ cannot satisfy the Nash equilibria condition (it is dominated by a pure-strategy). Now, by Assumption~\ref{as:pay}, $f_{\mech}^i(\cdot)$ is thrice continuously differentiable w.r.t. $\mech$, so the mapping $\arg\max_{\gamma\in\CA_s^i}\int_{\CA_s^{-i}}\int_{\CA_s^i}f_{\mech}^i(x)\delta(x-\gamma)\ms_s^{-i}(x^{-i})dx$ is and thus the value $f_{\mech}^i(\msb_s)$ is thrice continuously differentiable w.r.t. $\mech$. Then, this implies that the functional $g_s^i(\ms_s^i(\mech))$ is thrice continuously differentiable by the smoothness properties of $f_{\mech}^i(\cdot)$ and $g_s^i(\cdot)$ in Assumptions~\ref{as:con},\ref{as:pay}.

\textbf{H4}: The algorithm parameters are taken to satisfy these conditions

\textbf{H5}: Our algorithm is isolated to the compact set $\mechspace$, eliminating the necessity of asymptotic objective function structure such as this requirement. 

\textbf{H6}: This is satisfied.

\textbf{H7}: Let $P_{\eta}$ be a parametrized probability measure over the parameter space $\mechspace$, defined by the Radon-Nikodym derivative 
$\frac{dP_{\eta}(\mech)}{d\mech} = \frac{\exp(-2L(\mech)/\eta^2)}{Z_{\eta}}$
where $\eta>0$ and $Z_{\eta} = \int_{\mechspace}\exp(-2L(\mech)/\eta^2)d\mech < \infty$ is a normalizing term. Then the unique weak limit $P = \lim_{\eta \to 0}P_{\eta}$ is given by an indicator on points in the set $\hat{\mechspace} := \{\mech : L(\mech) = 0\}$, i.e.,
$P(\mech) = \chi_{\mech\in\hat{\mechspace}} / Z_0$
with $\chi$ the indicator function and $Z_0 = \int_{\mechspace}\chi_{\mech\in\hat{\mechspace}}d\mech$.

\textbf{H8}: The sequence is trivially tight, since $P(\mech_k \in \mechspace) = 1 \, \forall k\in\nat$.

\subsection{Proof of Theorem~\ref{thm:dreq}}
\label{sec:pfdreq}
We have that $\otimes_{t=1}^T \msb_t \in B_{\epsilon}(\otimes_{t=1}^T\hat{\msb}_t)$. Let $\hat{\psi} = [\hu_1^1,\hlam_1^1,\dots]$ denote the $\arg\min$ corresponding to the minimization \eqref{eq:dro}. Then, since $\hat{L}_{DR}(\mech) = 0$, we must have
\begin{align*}
    0 = \CE_{\mech \sim \otimes_{t=1}^T\msb_t}[h(\hat{\psi},\mech)] = \int_{\otimes_{t=1}^T\CA_{\mech,t}}\left[\arg\min_{r\geq 0} : \hu_s^i - \hu_t^i - \hlam_t^i g_{\mech,t}^i(\gamma_s^i) \leq \hlam_t^i\,r \right]d\otimes_{t=1}^T \msb_t
\end{align*}
    with $\gamma_s^i \sim \msb_s^i$, and thus for all $s,t\in[T], i\in[\na]$:
\begin{align*}
    &\left[\arg\min_{r\geq 0} : \hu_s^i - \hu_t^i - \hlam_t^i g_{\mech,t}^i(\gamma_s^i) \leq \hlam_t^i\,r\right] = 0\,\,\, \forall \gamma_s^i \in \textrm{supp}(\msb_s^i) \\
    &\Rightarrow \left[\arg\min_{r\geq 0} : \hu_s^i - \hu_t^i - \hlam_t^i \min_{\gamma \in \textrm{supp}(\msb_s^i)}g_{\mech,t}^i(\gamma) \leq \hlam_t^i\,r\right] = 0 \\
    &\Rightarrow \left[\arg\min_{r\geq 0} : \hu_s^i - \hu_t^i - \hlam_t^i \int_{\CA_s^i}g_{\mech,t}^i(\gamma_s^i)d\msb_s^i \leq \hlam_t^i\,r \right] = 0 
\end{align*}

\section{Generalized Revealed Preferences by Mixed-Integer Linear Programming \cite{snow2022identifying}}

{\color{black} Consider a group composed of $M$ members, each of which can consume some quantity of $N$ goods. At time $t \in \nat$ each member $i \in \{1,\dots, M\}$ has consumption given by the vector $\consi \in \posreals^N$, and the aggregate group consumption is given by $\cons = \sum_{i=1}^{\na}\consi \in \posreals^N$, subject to price vector (probe) $\alpha_{\dtime} \in \posreals^N$. It is assumed that the preferences of each member $i$ can be represented by a non-satiated and non-decreasing utility function $\utili(\beta),\ \beta \in \posreals^N$. 

For each $\dtime \in \{1,\dots,T\}$, suppose an analyst observes probe signals $\pricet \in \reals^N_+$, aggregate consumption $\cons 
 = \sum_{i=1}^M\consi \in \reals^N$, and "assignable quantities" $\asconsi \leq \consi \ \forall i \in \{1,\dots,M\}$. The assignable quantity $\asconsi$ represents observed consumption of some subset of the quantities consumed by individual $i$ at time $t$, and hence it is elementwise no greater than the total consumption vector $\consi$. The apparatus by which these assignable quantities are observed, and the amount observed, varies by application. We denote this dataset as 
 \begin{equation}
 \label{dataset}
    \dataset = \{\pricet, \cons, \{\asconsi\}_{i=1}^M, \ t \in \{1,\dots,T\} \}
 \end{equation}
 We emphasize that the true individual consumption vectors $\consi$ are hidden. Given this dataset, \cite{cherchye2011revealed} provides necessary and sufficient conditions for consistency with Pareto-optimality. To formulate these conditions, we need one more definition:
 For each observation $t$, \textit{feasible personalized quantities} $\feasconsi \in \posreals^N, \ i=1,\dots,M$ satisfy $\feasconsi \geq \asconsi \ \forall i$ and $\sum_{i=1}^{\na}\feasconsi = \cons$.

 Now we provide the main result of \cite{cherchye2011revealed} which states the equivalence between a consistency of dataset $\dataset$ with coordination and existence of a non-empty feasible region of a set of inequalities.

 \begin{theorem}{\cite{cherchye2011revealed}}
 \label{thm:cherchye1}
    Let $\dataset$ be a set of observations. The following are equivalent:
    \begin{enumerate}
    \item there exist a set of $M$ concave and continuous utility functions $U^1,\dots,U^m$ such that $\dataset$ is consistent with Pareto-optimality.
    \item there exist feasible personalized quantities $\feasconsi$ and numbers $u_j^i > 0, \lambda_j^i > 0$ such that for all $s,t \in \{1,\dots,T\}$: 
    \begin{equation}
    \label{af_ineq}
        u_s^i - u_t^i \leq \lambda_t^i[\alpha_t'q_s^i - \alpha_t 'q_t^i]
    \end{equation}
    for each member $i=1,\dots,M$.
    \end{enumerate}
\end{theorem}
\begin{proof}
See Appendix A.1 of \cite{cherchye2011revealed}.
\end{proof}
    \begin{lemma}
    \label{lem:util}
    Suppose the conditions of Theorem~\ref{thm:cherchye1} hold. Explicit monotone and concave utility functions that rationalize the dataset by satisfying Pareto optimality are given by
            \begin{equation}
            \label{eq:util}
                \hat{U}^i(\beta) = \min_{t \in \{1,\dots,T\}} \{u_t^i + \lambda_t^i \alpha_t'(\beta - q_t^i)\}
            \end{equation}
   \end{lemma}
 \begin{proof}
     This follows immediately by application of Afriat's Theorem \cite{afriat1967construction}.
 \end{proof}
 
 Thus if the feasible region of \eqref{af_ineq} is nonempty, then monotone concave utility functions which rationalize the dataset can be reconstructed using \eqref{eq:util}.

 Implementing this feasibility test \eqref{af_ineq} in practice presents a problem however, because the set $\{\feasconsi\}_{i=1}^M$ is \textit{unobserved}. To overcome this, \cite{cherchye2011revealed} present an equivalent mixed-integer linear program (MILP) formulation which allows for practical computation of this coordination test.

 \begin{theorem}{\cite{cherchye2011revealed}}
 \label{MILP}
 Let $\dataset$ be a set of observations. There exists a set of $\na$concave and continuous utility functions $U^1,\dots,U^M$ such that $\dataset$ is consistent with coordination if and only if there exist $\feasconsi \in \posreals^N$, $\eta_t^i \in \posreals$, and $x_{st}^i \in \{0,1\}, i=1,\dots,M$ that satisfy
 \begin{enumerate}[label=\roman*)]
    \item $\sum_{i=1}^M\feasconsi = \cons$ and $\feasconsi \geq \asconsi$
    \item $\eta_t^i = \pricet'\feasconsi$
    \item $\eta_s^i - \pricet'\feasconsi < y_s x_{st}^i$
    \item $x_{su}^i + x_{ut}^i \leq 1 + x_{st}^i$
    \item $\eta_t^i - \pricet'\feasconsi \leq y_t(1-x_{st}^i)$
 \end{enumerate}
 where $y_t = \pricet'\cons$ is the group consumption cost at time $t$.

 \end{theorem}

 \begin{proof}
    See Proposition 2 of \cite{cherchye2011revealed}.
 \end{proof}
 
 Theorem~\ref{MILP} provides a generalization of Afriat's theorem in revealed preferences, to handle assignable quantities via mixed-integer linear programming. An immediate extension of the work in this paper, would be to generalize Theorem~\ref{thm:MA_Af} in such a way; this would extend and tie in our results on finite-sample robustness in Section~\ref{sec:robustrp}.}

\end{document}